\documentclass[review]{elsarticle}

\usepackage{lineno,hyperref}
\usepackage{latexsym}
\usepackage{mathtools}
\usepackage{enumerate}
\usepackage{extarrows}
\usepackage{xcolor}
\usepackage{cases}
\usepackage{array}
\usepackage{url}
\usepackage{amsmath}
\usepackage{mathrsfs}
\usepackage{amsthm}
\usepackage{amsfonts}
\usepackage{amssymb}
\usepackage{epstopdf}
\usepackage{graphicx}
\usepackage{subfigure}
\usepackage{algorithm}
\usepackage{algorithmic}
\usepackage{threeparttable}
\usepackage{multirow}
\usepackage{booktabs}

\def\bdiag{\mathop{\rm bdiag}}
\def\bcirc{\mathop{\rm bcirc}}
\def\supp{\mathop{\rm supp}}

\def\supp{\mathop{\rm supp}}
\def\rank{\mathop{\rm rank}}
\def\fold{\mathop{\rm fold}}
\def\unfold{\mathop{\rm unfold}}
\newtheorem{theorem}{Theorem}[section]

\newtheorem{definition}[theorem]{Definition} % Let definition show the same serial number with theorem

\newtheorem{lemma}[theorem]{Lemma} % Let lemma show the same serial number with theorem

\newtheorem{remark}[theorem]{Remark} % Let remark show the same serial number with theorem

\newtheorem{proposition}[theorem]{Proposition}  % Let proposition show the same serial number with theorem

\newtheorem{corollary}[theorem]{Corollary}
\modulolinenumbers[5]

\journal{Journal of \LaTeX\ Templates}

\begin{document}

\begin{frontmatter}

\title{RIP-based performance guarantee for low-tubal-rank tensor recovery}


\author[mymainaddress]{Feng Zhang}
\ead{zhangf@email.swu.edu.cn}

\author[mymainaddress]{Wendong Wang}
\ead{d.sylan@foxmail.com}

\author[mymainaddress]{Jianwen Huang}
\ead{hjw1303987297@126.com}

\author[mythirdaddress]{Yao Wang}
\ead{yao.s.wang@gmail.com}

\author[mysecondaddress]{Jianjun Wang\corref{mycorrespondingauthor}}
\cortext[mycorrespondingauthor]{Corresponding author}
\ead{wjj@swu.edu.cn}

\address[mymainaddress]{School of Mathematics and Statistics, Southwest University, Chongqing, 400715, China}
\address[mysecondaddress]{College of Artificial Intelligence, Southwest University, Chongqing, 400715, China}
\address[mythirdaddress]{School of Mathematics and Statistics, Xi'an Jiaotong University, Xi'an, 710049, China}

\begin{abstract}
The essential task of tensor data analysis focuses on the tensor decomposition and the corresponding notion of rank. In this paper, by introducing the notion of tensor Singular Value Decomposition (t-SVD), we establish a Regularized Tensor Nuclear Norm Minimization (RTNNM) model for low-tubal-rank tensor recovery. As we know that many variants of the Restricted Isometry Property (RIP) have proven to be crucial analysis tools for sparse recovery. In the t-SVD framework, we initiatively define a novel tensor Restricted Isometry Property (t-RIP). Furthermore, we show that any third-order tensor $\boldsymbol{\mathcal{X}}$ can stably be recovered from few linear noise measurements under some certain t-RIP conditions via the RTNNM model. We note that, as far as the authors are
aware, such kind of result has not previously been reported in the literature.
\end{abstract}


\end{frontmatter}

%\linenumbers

\section{Introduction}
\label{Introduction}
Utilizing the tensor model, possessed of the ability to make full use of multi-linear structure, instead of the traditional matrix-based model to analyze multi-dimensional data (tensor data) has widely attracted attention. Indeed, in the real world, the carriers of information processing are more commonly tensor data containing more complex information and structures, such as video, hyperspectral images and communication networks stream data, etc. Low-Rank Tensor Recovery (LRTR) as a representative problem is not only a mathematical natural generalization of the Compressed Sensing (CS) and Low-Rank Matrix Recovery (LRMR) problem, but also there exists lots of reconstruction applications of data that have intrinsically many dimensions in the context of LRTR including signal processing \cite{Liu2013Tensor,Wang2019Generalized}, machine learning \cite{Romera2013Multilinear,wang2017compressive}, data mining \cite{M2011Applications,friedland2014compressive}, and many others \cite{Lu2019Tensor,rauhut2017low,Zhang2014Novel,Lu2018Exact,shi2013guarantees}.

The purpose of LRTR is to reconstruct a low-rank tensor $\boldsymbol{\mathcal{X}}\in\mathbb{R}^{n_{1} \times n_{2} \times n_{3}}$ (this article considers only the third-order tensor without loss of generality) from linear noise measurements $\boldsymbol{y}=\boldsymbol{\mathfrak{M}}(\boldsymbol{\mathcal{X}})+\boldsymbol{w}$, where $\boldsymbol{\mathfrak{M}}: \mathbb{R}^{n_{1} \times n_{2} \times n_{3}} \rightarrow \mathbb{R}^{m}$ $(m\ll n_{1}n_{2}n_{3})$ is a random map with i.i.d. Gaussian entries and $\boldsymbol{w}\in\mathbb{R}^{m}$ is a vector of measurement errors. To be specific, we consider addressing the following rank minimization problem
\begin{equation}\label{rank min}
    \min \limits_{{\boldsymbol{\mathcal{X}}}\in\mathbb{R}^{n_{1} \times n_{2} \times n_{3}}}~\rank(\boldsymbol{\mathcal{X}}),~~s.t.~~\|\boldsymbol{y}-\boldsymbol{\mathfrak{M}}(\boldsymbol{\mathcal{X}})\|_{2}\leq\epsilon,
\end{equation}
where $\epsilon$ is a positive constant. However, there exist two tough problems in dealing with model (\ref{rank min}). First, the rank minimization problem (\ref{rank min}) is NP-hard and non-convex. Second, how should the rank of a tensor be defined? Unlike the unique definition of matrix rank, there exist different notions of tensor rank induced by different tensor decompositions. For example, the CANDECOMP/PARAFAC (CP) rank is defined as the minimum number of rank-1 tensors required to achieve the CP decomposition \cite{kiers2000towards}; The Tucker rank induced by the Tucker decomposition \cite{Tucker1966Some} is defined as a vector whose $i$-th element is the rank of the mode-$i$ unfolding matrix of a tensor. However, computing the CP rank of a tensor is already NP-hard in general and its convex relaxation is also intractable. Although the Tucker rank is tractable, its convex surrogate, Sum of Nuclear Norms (SNN) as  \cite{Liu2013Tensor}, is not the tightest convex relaxation, which will lead to that the surrogate model can be substantially suboptimal \cite{mu2014square}.

Recently, based on the definition of tensor Singular Value Decomposition (t-SVD) \cite{Kilmer2011Factorization,Martin2013An} that enjoys many similar properties as the matrix case, the tensor tubal rank (see Definition \ref{Tensor tubal rank}) is proposed by Kilmer et al. \cite{kilmer2013third}. Along those lines, Lu et al. \cite{Lu2019Tensor} given a new and rigorous way to define the tensor average rank (see Definition \ref{Tensor average rank}) and the tensor nuclear norm (see Definition \ref{Tensor nuclear norm}), and proved that the convex envelop of tensor average rank is tensor nuclear norm within the unit ball of the tensor spectral norm. Furthermore, they pointed out that a tensor always has low average rank if it has low tubal rank. Therefore, a convex tensor nuclear norm minimization (TNNM) model based on the assumption of low tubal rank for tensor recovery has been proposed in \cite{Lu2019Tensor}, which solves
\begin{equation}\label{tensor nuclear norm min}
  \min \limits_{{\boldsymbol{\mathcal{X}}}\in\mathbb{R}^{n_{1} \times n_{2} \times n_{3}}}~\|\boldsymbol{\mathcal{X}}\|_{*},~~s.t.~~\|\boldsymbol{y}-\boldsymbol{\mathfrak{M}}(\boldsymbol{\mathcal{X}})\|_{2}\leq\epsilon,
\end{equation}
where tensor nuclear norm $\|\cdot\|_{*}$ is as the convex surrogate of tensor average rank.

In order to facilitate the design of algorithms and the needs of practical applications, instead of considering the constrained-TNNM (\ref{tensor nuclear norm min}), in this paper, we present a theoretical analysis for Regularized Tensor Nuclear Norm Minimization (RTNNM) model, which takes the form
\begin{equation}\label{tensor nuclear norm min unconstrained}
  \min \limits_{{\boldsymbol{\mathcal{X}}}\in\mathbb{R}^{n_{1} \times n_{2} \times n_{3}}}~\|\boldsymbol{\mathcal{X}}\|_{*}+\frac{1}{2\lambda}\|\boldsymbol{y}-\boldsymbol{\mathfrak{M}}(\boldsymbol{\mathcal{X}})\|_{2}^{2},
\end{equation}
where $\lambda$ is a positive parameter. According to \cite{zhang2016one}, there exists a $\lambda > 0$ such that the solution to the regularization problem (\ref{tensor nuclear norm min unconstrained}) also solves the constrained problem (\ref{tensor nuclear norm min}) for any $\epsilon > 0$, and vice versa. However, model (\ref{tensor nuclear norm min unconstrained}) is more commonly used than model (\ref{tensor nuclear norm min}) when the noise level is not given or cannot be accurately estimated. There are many examples of solving RTNNM problem (\ref{tensor nuclear norm min unconstrained}) based on the tensor nuclear norm heuristic. For instance, the low-tubal-rank tensor completion problem \cite{Lu2018Exact} can be thought of as partial observations under the identity transformation and solved as a special case of the RTNNM problem (\ref{tensor nuclear norm min unconstrained}) in an alternating direction method of multipliers (ADMM) \cite{boyd2011distributed} framework. While the application and algorithm research of (\ref{tensor nuclear norm min unconstrained}) is already well-developed, only a few contributions on the theoretical results with regard to performance guarantee for low-tubal-rank tensor recovery are available so far. The Restricted Isometry Property (RIP) introduced by Cand\`{e}s et al. \cite{candes2005Decoding} and Recht et al. \cite{recht2010guaranteed} is one of the most widely used frameworks in CS and LRMR. In this paper, we generalize the RIP tool to tensor case based on t-SVD and establish the guaranteed conditions for low-tubal-rank tensor recovery.

As we know, different tensor decompositions induce different notions of tensor rank, which will lead to different notions of the tensor RIP. For example, in 2013, Shi et al. \cite{shi2013guarantees} defined the first tensor RIP based on Tucker decomposition and showed that a tensor $\boldsymbol{\mathcal{X}}\in\mathbb{R}^{n_{1} \times n_{2} \times n_{3}}$ with Tucker rank-$(r_{1}, r_{2}, r_{3})$ can be exactly reconstructed in the noiseless case if the linear map $\boldsymbol{\mathfrak{M}}$ satisfies the tensor RIP with the constant $\delta_{\Lambda}<0.4931$ for $\Lambda\in\left\{(2r_{1}, n_{2}, n_{3}),(n_{1},2r_{2}, n_{3}),(n_{1}, n_{2}, 2r_{3})\right\}$; In 2017, based on three variants of Tucker decomposition, i.e., the Higher Order Singular Value Decomposition (HOSVD), the Tensor Train format (TT), and the general Hierarchical Tucker decomposition (HT), Rauhut et al. \cite{rauhut2017low} also induced three notions of the tensor RIP. These tensor RIP definitions are difficult to implement due to relying on a rank tuple that differs greatly from the known matrix rank definition, which would result in some existing analysis tools and techniques not being available for tensor cases. In contrast, the tensor tubal rank induced by the t-SVD is consistent with the matrix rank in terms of form and properties. So, we initiatively define a novel tensor Restricted Isometry Property as follows:
\begin{definition}\quad(t-RIP)\label{Tensor RIP Definition}
A linear map $\boldsymbol{\mathfrak{M}}: \mathbb{R}^{n_{1} \times n_{2} \times n_{3}}\rightarrow \mathbb{R}^{m}$, is said to satisfy the t-RIP of order $r$ with tensor Restricted Isometry Constant (t-RIC) $\delta_{r}^{\boldsymbol{\mathfrak{M}}}$ if $\delta_{r}^{\boldsymbol{\mathfrak{M}}}$ is the smallest value $\delta^{\boldsymbol{\mathfrak{M}}}\in(0,1)$ such that
\begin{equation}\label{t-RIP}
(1-\delta^{\boldsymbol{\mathfrak{M}}})\|\boldsymbol{\mathcal{X}}\|_{F}^{2}\leq\|\boldsymbol{\mathfrak{M}}(\boldsymbol{\mathcal{X}})\|_{2}^{2}\leq(1+\delta^{\boldsymbol{\mathfrak{M}}})\|\boldsymbol{\mathcal{X}}\|_{F}^{2}
\end{equation}
holds for all tensors $\boldsymbol{\mathcal{X}}\in\mathbb{R}^{n_{1} \times n_{2} \times n_{3}}$ whose tubal rank is at most $r$.
\end{definition}
Note that the above definition is the natural generalization of the RIP from vectors and matrices to tensors. However, similar to LRMR and different from CS, One cannot view our t-RIC for low-tubal-rank tensor as the condition that all sub-tensors of $\boldsymbol{\mathfrak{M}}$ of a given size are well conditioned. Surprisingly, we are still able to derive analogous recovery results for the LRTR problem if the t-RIC of $\boldsymbol{\mathfrak{M}}$ satisfies certain conditions. As pointed out by an anonymous reviewer, we should explain why we assume that $\boldsymbol{\mathfrak{M}}$ contains i.i.d. Gaussian entries. That's because in our previous work \cite{zhang2019tensor} we have proved that the random Gaussian measurement operator $\boldsymbol{\mathfrak{M}}$ satisfies a t-RIP at tubal rank $r$ with high probability.

Next, let us review some existing theoretical results that will provide us with some inspiration. For constrained CS and LRMR, different conditions on the RIC have been introduced and studied in the literature \cite{candes2008restricted,Candes2011Tight,Cai2013Sharp}, etc. Among these sufficient conditions, especially, Cai and Zhang \cite{Cai2013Sparse} showed that for any given $t\geq4/3$, the RIC $\delta_{tr}^{\boldsymbol{M}}<\sqrt{\frac{t-1}{t}}$ for the vector case is sharp (the RIC $\delta_{tr}^{\boldsymbol{\mathcal{M}}}<\sqrt{\frac{t-1}{t}}$ for the matrix case) ensures the exact recovery in the noiseless case and stable recovery in the noisy case for $r$-sparse signals (matrices with rank at most $r$). In addition, Zhang and Li \cite{Zhang2018A} obtained another part of the sharp condition, that is $\delta_{tr}^{\boldsymbol{M}}<\frac{t}{4-t} (\delta_{tr}^{\boldsymbol{\mathcal{M}}}<\frac{t}{4-t})$ with $0<t<4/3$. The results mentioned above are currently the best in the field. In view of unconstrained sparse vector recovery, as far as we know that Zhu \cite{Zhu2008Stable} first studied this kind of problem in 2008 and he pointed out that $r$-sparse signals can be recovered stably if $\delta_{4r}^{\boldsymbol{M}}+2\delta_{5r}^{\boldsymbol{M}}<1$. Next, in 2015, Shen et al. \cite{Shen2015Stable} got a sufficient condition $\delta_{2r}^{\boldsymbol{M}}<0.2$ under redundant tight frames. Recently, Ge et al. \cite{Ge2018Stable} proved that if the noisy vector $\boldsymbol{w}$ satisfies the $\ell_{\infty}$ bounded noise constraint (i.e., $\|\boldsymbol{M}^{*}\boldsymbol{w}\|_{\infty}\leq\lambda/2$) and $\delta_{tr}^{\boldsymbol{M}}<\sqrt{\frac{t-1}{t+8}}$ with $t>1$, then $r$-sparse signals can be stably recovered. For unconstrained LRMR, the sufficient condition $\delta_{tr}^{\boldsymbol{\mathcal{M}}}<\sqrt{(t-1)/t}$ for $t>1$  and error upper bound estimation have been derived in our previous work \cite{wang2019low}.

Equipped with the t-RIP, in this paper, we aim to construct sufficient conditions for stable low-tubal-rank tensor recovery and obtain an upper bound estimate of error via solving (\ref{tensor nuclear norm min unconstrained}). The rest of the paper is organized as follows. In Section \ref{Notations and Preliminaries}, we introduce some notations and definitions. In Section \ref{Some Key Lemmas}, we give some key lemmas. In Section \ref{Main Result}, our main result is presented. In Section \ref{Numerical experiments}, some numerical experiments are conducted to support our analysis. The conclusion is addressed in Section \ref{Conclusion}. Finally, \ref{Appendix A} provides the proof of Lemma \ref{Lemma-2}.

\section{Notations and Preliminaries}
\label{Notations and Preliminaries}
We use lowercase letters for the entries, e.g. $x$, boldface letters for vectors, e.g. $\boldsymbol{x}$, capitalized boldface letters for matrices, e.g. $\boldsymbol{X}$ and capitalized boldface calligraphic letters for tensors, e.g. $\boldsymbol{\mathcal{X}}$. For a third-order tensor $\boldsymbol{\mathcal{X}}$, $\boldsymbol{\mathcal{X}}(i,:,:)$, $\boldsymbol{\mathcal{X}}(:,i,:)$ and $\boldsymbol{\mathcal{X}}(:,:,i)$ are used to represent the $i$th horizontal, lateral, and frontal slice. The frontal slice $\boldsymbol{\mathcal{X}}(:,:,i)$ can also be denoted as $\boldsymbol{X}^{(i)}$. The tube is denoted as $\boldsymbol{\mathcal{X}}(i,j,:)$. We denote the Frobenius norm as $\|\boldsymbol{\mathcal{X}}\|_{F}=\sqrt{\sum_{ijk}|x_{ijk}|^{2}}$. Defining some norms of matrix is also necessary. We denote by $\|\boldsymbol{X}\|_{F}=\sqrt{\sum_{ij}|x_{ij}|^{2}}=\sqrt{\sum_{i}\sigma_{i}^{2}(\boldsymbol{X})}$ the Frobenius norm of $\boldsymbol{X}$ and denote by $\|\boldsymbol{X}\|_{*}=\sum_{i}\sigma_{i}(\boldsymbol{X})$ the nuclear norm of $\boldsymbol{X}$, where $\sigma_{i}(\boldsymbol{X})$'s are the singular values of $\boldsymbol{X}$ and $\sigma(\boldsymbol{X})$ represents the singular value vector of matrix $\boldsymbol{X}$. Given a positive integer $\kappa$, we denote $[\kappa]=\{1,2,\cdots,\kappa\}$ and $\Gamma^{c}=[\kappa]\setminus\Gamma$ for any $\Gamma\subset[\kappa]$. The set of indices of the nonzero entries of a vector $\boldsymbol{v}$ is called the support of $\boldsymbol{v}$ and denoted as $\supp(\boldsymbol{v})$. $|\Gamma|$ is the cardinality for the index set.

For a third-order tensor $\boldsymbol{\mathcal{X}}$, let $\boldsymbol{\bar{\mathcal{X}}}$ be the discrete Fourier transform (DFT) along the third dimension of $\boldsymbol{\mathcal{X}}$, i.e., $\boldsymbol{\bar{\mathcal{X}}}=\rm fft(\boldsymbol{\mathcal{X}},[],3)$. Similarly, $\boldsymbol{\mathcal{X}}$ can be calculated from $\boldsymbol{\bar{\mathcal{X}}}$ by $\boldsymbol{\mathcal{X}}=\rm ifft(\boldsymbol{\bar{\mathcal{X}}},[],3)$. Let $\boldsymbol{\bar{X}}\in\mathbb{R}^{n_{1}n_{3}\times n_{2}n_{3}}$ be the block diagonal matrix with each block on diagonal as the frontal slice $\boldsymbol{\bar{X}}^{(i)}$ of $\boldsymbol{\bar{\mathcal{X}}}$, i.e.,
\begin{equation*}
\boldsymbol{\bar{X}}=\bdiag(\boldsymbol{\bar{\mathcal{X}}})=\left(
                                                                    \begin{array}{cccc}
                                                                      \boldsymbol{\bar{X}}^{(1)} &  &  &  \\
                                                                       & \boldsymbol{\bar{X}}^{(2)} &  &  \\
                                                                       &  & \ddots &  \\
                                                                       &  &  & \boldsymbol{\bar{X}}^{(n_{3})} \\
                                                                    \end{array}
                                                                  \right),
\end{equation*}
and $\bcirc(\boldsymbol{\mathcal{X}})\in\mathbb{R}^{n_{1}n_{3}\times n_{2}n_{3}}$ be the block circular matrix, i.e.,                                                                                                                           \begin{equation*}
\bcirc(\boldsymbol{\mathcal{X}})=\left(
                                                                    \begin{array}{cccc}
                                                                      \boldsymbol{X}^{(1)} & \boldsymbol{X}^{(n_{3})} & \cdots & \boldsymbol{X}^{(2)} \\
                                                                      \boldsymbol{X}^{(2)} & \boldsymbol{X}^{(1)} & \cdots & \boldsymbol{X}^{(3)} \\
                                                                      \vdots & \vdots & \ddots & \vdots \\
                                                                      \boldsymbol{X}^{(n_{3})} & \boldsymbol{X}^{(n_{3}-1)} & \cdots & \boldsymbol{X}^{(1)} \\
                                                                    \end{array}
                                                                  \right).
\end{equation*}
The $\rm unfold$ operator and its inverse operator $\rm fold$ are, respectively, defined as
\begin{equation*}
\unfold(\boldsymbol{\mathcal{X}})=\left(
  \begin{array}{cccc}
    \boldsymbol{X}^{(1)} & \boldsymbol{X}^{(2)} & \cdots & \boldsymbol{X}^{(n_{3})} \\
  \end{array}
\right)^{T},~~
\fold(\unfold(\boldsymbol{\mathcal{X}}))=\boldsymbol{\mathcal{X}}.
\end{equation*}
Then tensor-tensor product (t-product) between two third-order tensors can be defined as follows.
\begin{definition}\quad(t-product \cite{Kilmer2011Factorization})\label{t-product}
For tensors $\boldsymbol{\mathcal{A}}\in\mathbb{R}^{n_{1}\times n_{2}\times n_{3}}$ and $\boldsymbol{\mathcal{B}}\in\mathbb{R}^{n_{2}\times n_{4}\times n_{3}}$, the t-product $\boldsymbol{\mathcal{A}}\ast\boldsymbol{\mathcal{B}}$ is defined to be a tensor of size $n_{1}\times n_{4}\times n_{3}$,
\begin{equation*}
\boldsymbol{\mathcal{A}}\ast\boldsymbol{\mathcal{B}}=\fold(\bcirc(\boldsymbol{\mathcal{A}})\cdot\unfold(\boldsymbol{\mathcal{B}})).
\end{equation*}
\end{definition}

\begin{definition}\quad(Conjugate transpose \cite{Kilmer2011Factorization})\label{Conjugate transpose}
The conjugate transpose of a tensor $\boldsymbol{\mathcal{X}}$ of size $n_{1}\times n_{2}\times n_{3}$ is the $n_{2}\times n_{1}\times n_{3}$ tensor $\boldsymbol{\mathcal{X}}^{*}$ obtained by conjugate transposing each of the frontal slice and then reversing the order of transposed frontal slices 2 through $n_{3}$.
\end{definition}

\begin{definition}\quad(Identity tensor \cite{Kilmer2011Factorization})\label{Identity tensor}
The identity tensor $\boldsymbol{\mathcal{I}}\in\mathbb{R}^{n\times n\times n_{3}}$ is the tensor whose first frontal slice is the $n\times n$ identity matrix, and other frontal slices are all zeros.
\end{definition}

\begin{definition}\quad(Orthogonal tensor \cite{Kilmer2011Factorization})\label{Orthogonal tensor)}
A tensor $\boldsymbol{\mathcal{Q}}\in\mathbb{R}^{n\times n\times n_{3}}$ is orthogonal if it satisfies
\begin{equation*}
  \boldsymbol{\mathcal{Q}}^{*}\ast\boldsymbol{\mathcal{Q}}=\boldsymbol{\mathcal{Q}}\ast\boldsymbol{\mathcal{Q}}^{*}=\boldsymbol{\mathcal{I}}.
\end{equation*}
\end{definition}

\begin{definition}\quad(F-diagonal tensor \cite{Kilmer2011Factorization})\label{F-diagonal tensor}
A tensor is called F-diagonal if each of its frontal slices is a diagonal matrix.
\end{definition}

\begin{theorem}\quad(t-SVD \cite{Kilmer2011Factorization})\label{t-SVD}
Let $\boldsymbol{\mathcal{X}}\in\mathbb{R}^{n_{1}\times n_{2}\times n_{3}}$, the t-SVD factorization of tensor $\boldsymbol{\mathcal{X}}$ is
\begin{equation*}
\boldsymbol{\mathcal{X}}=\boldsymbol{\mathcal{U}} \ast \boldsymbol{\mathcal{S}} \ast \boldsymbol{\mathcal{V}}^{*},
\end{equation*}
where $\boldsymbol{\mathcal{U}}\in\mathbb{R}^{n_{1}\times n_{1}\times n_{3}}$ and $\boldsymbol{\mathcal{V}}\in\mathbb{R}^{n_{2}\times n_{2}\times n_{3}}$ are orthogonal, $\boldsymbol{\mathcal{S}}\in\mathbb{R}^{n_{1}\times n_{2}\times n_{3}}$ is an F-diagonal tensor. Figure \ref{illustration of the t-SVD} illustrates the t-SVD factorization.
\end{theorem}

\begin{figure}[ht]
\begin{center}
{\includegraphics[width=0.9\textwidth]{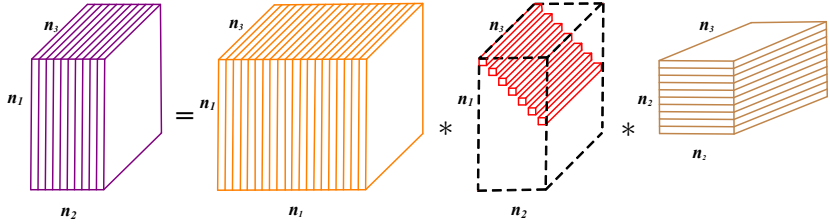}}
\end{center}
\vspace*{-14pt}
\caption{An illustration of the t-SVD of an $n_{1}\times n_{1}\times n_{3}$ tensor.}\label{illustration of the t-SVD}
\end{figure}

\begin{remark}\quad
For $\kappa=\min(n_{1},n_{2})$, the t-SVD of $\boldsymbol{\mathcal{X}}$ can be written
\begin{equation*}
\boldsymbol{\mathcal{X}}=\sum\nolimits_{i=1}^{\kappa}\boldsymbol{\mathcal{U}}_{\boldsymbol{\mathcal{X}}}(:,i,:) \ast \boldsymbol{\mathcal{S}}_{\boldsymbol{\mathcal{X}}}(i,i,:) \ast \boldsymbol{\mathcal{V}}_{\boldsymbol{\mathcal{X}}}(:,i,:)^{*}.
\end{equation*}
The diagonal vector of the first frontal slice of $\boldsymbol{\mathcal{S}}_{\boldsymbol{\mathcal{X}}}$ is denoted as $\boldsymbol{s}_{\boldsymbol{\mathcal{X}}}$. The best $r$-term approximation of $\boldsymbol{\mathcal{H}}$ with the tubal rank at most $r$ is denoted by
\begin{equation*}
\boldsymbol{\mathcal{X}}_{\max(r)}=\arg\min_{\rank_{t}(\boldsymbol{\tilde{\mathcal{X}}})\leq r}\|\boldsymbol{\mathcal{X}}-\boldsymbol{\tilde{\mathcal{X}}}\|_{F}=\sum\nolimits_{i=1}^{r}\boldsymbol{\mathcal{U}}_{\boldsymbol{\mathcal{X}}}(:,i,:) \ast \boldsymbol{\mathcal{S}}_{\boldsymbol{\mathcal{X}}}(i,i,:) \ast \boldsymbol{\mathcal{V}}_{\boldsymbol{\mathcal{X}}}(:,i,:)^{*},
\end{equation*}
and $\boldsymbol{\mathcal{X}}_{-\max(r)}=\boldsymbol{\mathcal{X}}-\boldsymbol{\mathcal{X}}_{\max(r)}$. In addition, for index set $\Gamma$, we have
\begin{equation*}
\boldsymbol{\mathcal{X}}_{\Gamma}=\sum\nolimits_{i\in\Gamma}\boldsymbol{\mathcal{U}}_{\boldsymbol{\mathcal{X}}}(:,i,:) \ast \boldsymbol{\mathcal{S}}_{\boldsymbol{\mathcal{X}}}(i,i,:) \ast \boldsymbol{\mathcal{V}}_{\boldsymbol{\mathcal{X}}}(:,i,:)^{*}.
\end{equation*}
\end{remark}

\begin{definition}\quad(Tensor tubal rank \cite{kilmer2013third})\label{Tensor tubal rank}
For $\boldsymbol{\mathcal{X}}\in\mathbb{R}^{n_{1}\times n_{2}\times n_{3}}$, the tensor tubal rank, denoted
as $\rank_{t}(\boldsymbol{\mathcal{X}})$, is defined as the number of nonzero singular tubes of $\boldsymbol{\mathcal{S}}$, where $\boldsymbol{\mathcal{S}}$ is from the t-SVD of $\boldsymbol{\mathcal{X}}=\boldsymbol{\mathcal{U}} \ast \boldsymbol{\mathcal{S}} \ast \boldsymbol{\mathcal{V}}^{*}$. We can write
\begin{equation*}
\rank\nolimits_{t}(\boldsymbol{\mathcal{X}})=\sharp\{i:\boldsymbol{\mathcal{S}}(i,i,:)\neq\boldsymbol{0}\}=\sharp\{i:\boldsymbol{\mathcal{S}}(i,i,1)\neq0\}.
\end{equation*}
\end{definition}
\begin{definition}\quad(Tensor average rank \cite{Lu2019Tensor})\label{Tensor average rank}
For $\boldsymbol{\mathcal{X}}\in\mathbb{R}^{n_{1}\times n_{2}\times n_{3}}$, the tensor average rank, denoted
as $\rank_{a}(\boldsymbol{\mathcal{X}})$, is defined as
\begin{equation*}
\rank\nolimits_{a}(\boldsymbol{\mathcal{X}})=\frac{1}{n_{3}}\rank(\bcirc(\boldsymbol{\mathcal{X}}))=\frac{1}{n_{3}}\rank(\bdiag(\boldsymbol{\bar{X}})).
\end{equation*}
\end{definition}
\begin{definition}\quad(Tensor nuclear norm \cite{Lu2019Tensor})\label{Tensor nuclear norm}
Let $\boldsymbol{\mathcal{X}}=\boldsymbol{\mathcal{U}} \ast \boldsymbol{\mathcal{S}} \ast \boldsymbol{\mathcal{V}}^{*}$ be the t-SVD of $\boldsymbol{\mathcal{X}}\in\mathbb{R}^{n_{1}\times n_{2}\times n_{3}}$. The tensor nuclear norm of $\boldsymbol{\mathcal{X}}$ is defined as $\|\boldsymbol{\mathcal{X}}\|_{*}:=\sum_{i=1}^{r}\boldsymbol{\mathcal{S}}(i,i,1)$, where $r=\rank_{t}(\boldsymbol{\mathcal{X}})$.
\end{definition}

\begin{proposition}\quad
For a third-order tensor $\boldsymbol{\mathcal{X}}$, we have the following properties
\begin{eqnarray}
\|\boldsymbol{\mathcal{X}}\|_{F}&=&\frac{1}{\sqrt{n_{3}}}\|\boldsymbol{\bar{X}}\|_{F},\label{Property 1}\\
\|\boldsymbol{\mathcal{X}}\|_{*}&=&\frac{1}{n_{3}}\|\boldsymbol{\bar{X}}\|_{*}.\label{Property 2}\\
\rank(\boldsymbol{\bar{X}})&\leq& n_{3}\rank\nolimits_{t}(\boldsymbol{\mathcal{X}}).\label{Property 3}
\end{eqnarray}
\end{proposition}

\section{Some Key Lemmas}
\label{Some Key Lemmas}
We present the following lemmas, which will play a key role in proving our sufficient conditions for low-tubal-rank tensor recovery.
\begin{lemma}\quad\cite{Cai2013Sparse}\quad\label{Sparse-Representation}
For a positive number $\phi$ and a positive integer $s$, define the polytope $T(\phi,s)\subset \mathbb{R}^{n}$ by
\begin{equation*}
T(\phi, s)=\{\boldsymbol{v}\in\mathbb{R}^{n}: \|\boldsymbol{v}\|_{\infty}\leq\phi, \|\boldsymbol{v}\|_{1}\leq s\phi\}.
\end{equation*}
For any $\boldsymbol{v}\in\mathbb{R}^{n}$, define the set of sparse vectors $U(\phi, s, \boldsymbol{v})\subset\mathbb{R}^{n}$ by
\begin{equation*}
U(\phi, s, \boldsymbol{v})=\{\boldsymbol{u}\in\mathbb{R}^{n}: \supp(\boldsymbol{u})\subseteq\supp(\boldsymbol{v}), \|\boldsymbol{u}\|_{0}\leq s, \|\boldsymbol{u}\|_{1}=\|\boldsymbol{v}\|_{1}, \|\boldsymbol{u}\|_{\infty}\leq\phi\}.
\end{equation*}
Then $\boldsymbol{v}\in T(\phi, s)$ if and only if $\boldsymbol{v}$ is in the convex hull of $U(\phi, s, \boldsymbol{v})$. In particular, any $\boldsymbol{v}\in T(\phi, s)$ can be expressed as
\begin{equation*}
\boldsymbol{v}=\sum_{i=1}^{N}\gamma_{i}\boldsymbol{u}_{i}
\end{equation*}
where $\boldsymbol{u}_{i}\in U(\phi, s, \boldsymbol{v})$ and $0\leq\gamma_{i}\leq1$, $\sum_{i=1}^{N}\gamma_{i}=1$.
\end{lemma}
This elementary technique introduced by T. Cai and A. Zhang \cite{Cai2013Sparse} shows that any point in a polytope can be represented as a convex combination of sparse vectors and makes the analysis surprisingly simple.

The following lemma shows that a suitable t-RIP condition implies the robust null space property \cite{Foucart2014stability} of the linear map $\boldsymbol{\mathfrak{M}}$.
\begin{lemma}\quad\label{Lemma-2}
Let the linear map $\boldsymbol{\mathfrak{M}}: \mathbb{R}^{n_{1}\times n_{2}\times n_{3}}\rightarrow\mathbb{R}^{n}$ satisfies the t-RIP of order $tr(t>1)$ with t-RIC $\delta_{tr}^{\boldsymbol{\mathfrak{M}}}\in(0,1)$. Then for any tensor $\boldsymbol{\mathcal{H}}\in\mathbb{R}^{n_{1}\times n_{2}\times n_{3}}$ and any subset $\Gamma\subset[\kappa]$ with $|\Gamma|=r$ and $\kappa=\min(n_{1},n_{2})$, it holds that
\begin{equation}\label{Lemma2-results}
\|\boldsymbol{\mathcal{H}}_{\Gamma}\|_{F}\leq \eta_{1}\|\boldsymbol{\mathfrak{M}}(\boldsymbol{\mathcal{H}})\|_{2}+\frac{\eta_{2}}{\sqrt{r}}\|\boldsymbol{\mathcal{H}}_{\Gamma^{c}}\|_{*},
\end{equation}
where
\begin{equation*}
\eta_{1}\triangleq\frac{2}{(1-\delta_{tr}^{\boldsymbol{\mathfrak{M}}})\sqrt{1+\delta_{tr}^{\boldsymbol{\mathfrak{M}}}}},~~\rm and~~ \eta_{2}\triangleq\frac{\sqrt{n_{3}}\delta_{tr}^{\boldsymbol{\mathfrak{M}}}}{\sqrt{(1-(\delta_{tr}^{\boldsymbol{\mathfrak{M}}})^{2})(t-1)}}.
\end{equation*}
\end{lemma}
\begin{proof}\quad
Please see \ref{Appendix A}.
\end{proof}

In order to prove the main result, we still need the following lemma.
\begin{lemma}\quad\label{Lemma-3}
If the noisy measurements $\boldsymbol{y}=\boldsymbol{\mathfrak{M}}(\boldsymbol{\mathcal{X}})+\boldsymbol{w}$ of tensor $\boldsymbol{\mathcal{X}}\in\mathbb{R}^{n_{1}\times n_{2}\times n_{3}}$ are observed with noise level $\|\boldsymbol{w}\|_{2}\leq\epsilon$, then for any subset $\Gamma\subset[\kappa]$ with $|\Gamma|=r$ and $\kappa=\min(n_{1},n_{2})$, the minimization solution
$\boldsymbol{\hat{\mathcal{X}}}$ of (\ref{tensor nuclear norm min unconstrained}) satisfies
\begin{equation}
\|\boldsymbol{\mathfrak{M}}(\boldsymbol{\mathcal{H}})\|_{2}^{2}-2\epsilon\|\boldsymbol{\mathfrak{M}}(\boldsymbol{\mathcal{H}})\|_{2}
\leq2\lambda(\|\boldsymbol{\mathcal{H}}_{\Gamma}\|_{*}-\|\boldsymbol{\mathcal{H}}_{\Gamma^{c}}\|_{*}+2\|\boldsymbol{\mathcal{X}}_{\Gamma^{c}}\|_{*})
,\label{Lemma3-results1}
\end{equation}
and
\begin{equation}\label{Lemma3-results2}
\|\boldsymbol{\mathcal{H}}_{\Gamma^{c}}\|_{*}\leq\|\boldsymbol{\mathcal{H}}_{\Gamma}\|_{*}+2\|\boldsymbol{\mathcal{X}}_{\Gamma^{c}}\|_{*}+\frac{\epsilon}{\lambda}\|\boldsymbol{\mathfrak{M}}(\boldsymbol{\mathcal{H}})\|_{2},
\end{equation}
where $\boldsymbol{\mathcal{H}}\triangleq \boldsymbol{\hat{\mathcal{X}}}-\boldsymbol{\mathcal{X}}$.
\end{lemma}
\begin{proof}\quad
Since $\boldsymbol{\hat{\mathcal{X}}}$ is the minimizer of (\ref{tensor nuclear norm min unconstrained}), we have
\begin{equation*}
\|\boldsymbol{\hat{\mathcal{X}}}\|_{*}+\frac{1}{2\lambda}\|\boldsymbol{y}-\boldsymbol{\mathfrak{M}}(\boldsymbol{\hat{\mathcal{X}}})\|_{2}^{2}\leq\|\boldsymbol{\mathcal{X}}\|_{*}+\frac{1}{2\lambda}\|\boldsymbol{y}-\boldsymbol{\mathfrak{M}}(\boldsymbol{\mathcal{X}})\|_{2}^{2}.
\end{equation*}
Also because $\boldsymbol{\hat{\mathcal{X}}}=\boldsymbol{\mathcal{H}}+\boldsymbol{\mathcal{X}}$ and $\boldsymbol{y}=\boldsymbol{\mathfrak{M}}(\boldsymbol{\mathcal{X}})+\boldsymbol{w}$, so the above inequality is equivalent to
\begin{equation*}
\|\boldsymbol{\mathfrak{M}}(\boldsymbol{\mathcal{H}})\|_{2}^{2}-2\langle \boldsymbol{w}, \boldsymbol{\mathfrak{M}}(\boldsymbol{\mathcal{H}})\rangle\leq 2\lambda(\|\boldsymbol{\mathcal{X}}\|_{*}-\|\boldsymbol{\hat{\mathcal{X}}}\|_{*}).
\end{equation*}
It follows from the Cauchy-Schwartz inequality and assumption $\|\boldsymbol{w}\|_{2}\leq \epsilon$ that
\begin{equation}\label{Estimate-L}
\|\boldsymbol{\mathfrak{M}}(\boldsymbol{\mathcal{H}})\|_{2}^{2}-2\langle \boldsymbol{w}, \boldsymbol{\mathfrak{M}}(\boldsymbol{\mathcal{H}})\rangle\geq\|\boldsymbol{\mathfrak{M}}(\boldsymbol{\mathcal{H}})\|_{2}^{2}-2\epsilon\|\boldsymbol{\mathfrak{M}}(\boldsymbol{\mathcal{H}})\|_{2}.
\end{equation}
On the other hand, we have
\begin{eqnarray}\label{Estimate-R}
 \|\boldsymbol{\hat{\mathcal{X}}}\|_{*}-\|\boldsymbol{\mathcal{X}}\|_{*}&=&\|(\boldsymbol{\mathcal{H}}+\boldsymbol{\mathcal{X}})_{\Gamma}\|_{*}+\|(\boldsymbol{\mathcal{H}}+\boldsymbol{\mathcal{X}})_{\Gamma^{c}}\|_{*}-(\|\boldsymbol{\mathcal{X}}_{\Gamma}\|_{*}+\|\boldsymbol{\mathcal{X}}_{\Gamma^{c}}\|_{*})\nonumber\\
&\geq&(\|\boldsymbol{\mathcal{X}}_{\Gamma}\|_{*}-\|\boldsymbol{\mathcal{H}}_{\Gamma}\|_{*})+(\|\boldsymbol{\mathcal{H}}_{\Gamma^{c}}|_{*}-\|\boldsymbol{\mathcal{X}}_{\Gamma^{c}}\|_{*})-(\|\boldsymbol{\mathcal{X}}_{\Gamma}\|_{*}+\|\boldsymbol{\mathcal{X}}_{\Gamma^{c}}\|_{*})\nonumber\\
&\geq&\|\boldsymbol{\mathcal{H}}_{\Gamma^{c}}|_{*}-\|\boldsymbol{\mathcal{H}}_{\Gamma}|_{*}-2\|\boldsymbol{\mathcal{X}}_{\Gamma^{c}}|_{*}.
\end{eqnarray}
Combining (\ref{Estimate-L}) and (\ref{Estimate-R}) and by a simple calculation, we get (\ref{Lemma3-results1}). As to (\ref{Lemma3-results2}), it is obtained by subtracting the term $\|\boldsymbol{\mathfrak{M}}(\boldsymbol{\mathcal{H}})\|_{2}^{2}$ from the left-hand side of (\ref{Lemma3-results1}).
\end{proof}

\section{Main Results}
\label{Main Result}
With preparations above, now we present our main result.
\begin{theorem}\quad\label{Theorem-1}
For any observed vector $\boldsymbol{y}=\boldsymbol{\mathfrak{M}}(\boldsymbol{\mathcal{X}})+\boldsymbol{w}$ of tensor $\boldsymbol{\mathcal{X}}\in\mathbb{R}^{n_{1}\times n_{2}\times n_{3}}$ corrupted by an unknown noise $\boldsymbol{w}$, with bounded constrain $\|\boldsymbol{w}\|_{2}\leq\epsilon$, if $\boldsymbol{\mathfrak{M}}$ satisfies t-RIP with
\begin{equation}\label{Th1_condition}
\delta_{tr}^{\boldsymbol{\mathfrak{M}}}<\sqrt{\frac{t-1}{n_{3}^{2}+t-1}}
\end{equation}
for certain $t>1$, then we have
\begin{equation}
\|\boldsymbol{\mathfrak{M}}(\boldsymbol{\hat{\mathcal{X}}}-\boldsymbol{\mathcal{X}})\|_{2}\leq C_{1}\|\boldsymbol{\mathcal{X}}_{-\max(r)}\|_{*}+ C_{2},\label{Theorem-1-results1}
\end{equation}
and
\begin{equation}\label{Theorem-1-results2}
\|\boldsymbol{\hat{\mathcal{X}}}-\boldsymbol{\mathcal{X}}\|_{F}\leq C_{3}\|\boldsymbol{\mathcal{X}}_{-\max(r)}\|_{*}+ C_{4},
\end{equation}
where $\boldsymbol{\hat{\mathcal{X}}}$ is the solution to (\ref{tensor nuclear norm min unconstrained}), and $C_{i}, i=1,2,3,4$ are denoted as
\begin{eqnarray*}
C_{1}&=&\frac{2}{\sqrt{r}\eta_{1}},\quad C_{2}=2\sqrt{r}\eta_{1}\lambda+2\epsilon,\\
C_{3}&=&\frac{2\sqrt{r}\eta_{1}(2\sqrt{n_{3}r}+1+\eta_{2})\lambda+2(\sqrt{n_{3}r}+\eta_{2})\epsilon}{r\eta_{1}(1-\eta_{2})\lambda},\\
C_{4}&=&\frac{(\sqrt{n_{3}r}+1)\eta_{1}\lambda+(\sqrt{n_{3}r}-\sqrt{n_{3}}\eta_{2}+\sqrt{n_{3}}+1)\epsilon}{(1-\eta_{2})\lambda(2\sqrt{r}\eta_{1}\lambda+2\epsilon)^{-1}}.
\end{eqnarray*}
\end{theorem}

\begin{proof}\quad
For convenience, let
\begin{equation*}
T=\supp(\boldsymbol{s}_{\boldsymbol{\mathcal{X}}_{\max(r)}})
\end{equation*}
be an index set with cardinality $|T|\leq r$. In addition, if we set $\boldsymbol{\mathcal{H}}=\boldsymbol{\hat{\mathcal{X}}}-\boldsymbol{\mathcal{X}}$ and $\rank(\bdiag(\boldsymbol{\bar{\mathcal{H}}}_{\Gamma}))=\bar{r}$, then by inequality (\ref{Lemma2-results}) and (\ref{Lemma3-results1}), we would get
\begin{eqnarray*}
\|\boldsymbol{\mathfrak{M}}(\boldsymbol{\mathcal{H}})\|_{2}^{2}-2\epsilon\|\boldsymbol{\mathfrak{M}}(\boldsymbol{\mathcal{H}})\|_{2}&\leq&2\lambda(\|\boldsymbol{\mathcal{H}}_{\Gamma}\|_{*}-\|\boldsymbol{\mathcal{H}}_{\Gamma^{c}}\|_{*}+2\|\boldsymbol{\mathcal{X}}_{\Gamma^{c}}\|_{*})\\
&\stackrel{(a)}{=}&2\lambda\left(\frac{1}{n_{3}}\|\bdiag(\boldsymbol{\bar{\mathcal{H}}}_{\Gamma})\|_{*}-\|\boldsymbol{\mathcal{H}}_{\Gamma^{c}}\|_{*}+2\|\boldsymbol{\mathcal{X}}_{\Gamma^{c}}\|_{*}\right)\\
&\leq&2\lambda\left(\frac{\sqrt{\bar{r}}}{n_{3}}\|\bdiag(\boldsymbol{\bar{\mathcal{H}}}_{\Gamma})\|_{F}-\|\boldsymbol{\mathcal{H}}_{\Gamma^{c}}\|_{*}+2\|\boldsymbol{\mathcal{X}}_{\Gamma^{c}}\|_{*}\right)\\
&\stackrel{(b)}{\leq}&2\lambda\left(\sqrt{r}\|\boldsymbol{\mathcal{H}}_{\Gamma}\|_{F}-\|\boldsymbol{\mathcal{H}}_{\Gamma^{c}}\|_{*}+2\|\boldsymbol{\mathcal{X}}_{\Gamma^{c}}\|_{*}\right)\\
&\leq&2\sqrt{r}\lambda\left(\eta_{1}\|\boldsymbol{\mathfrak{M}}(\boldsymbol{\mathcal{H}})\|_{2}+\frac{\eta_{2}}{\sqrt{r}}\|\boldsymbol{\mathcal{H}}_{\Gamma^{c}}\|_{*}\right)-2\lambda\|\boldsymbol{\mathcal{H}}_{\Gamma^{c}}\|_{*}+4\lambda\|\boldsymbol{\mathcal{X}}_{\Gamma^{c}}\|_{*}\\
&=&2\sqrt{r}\eta_{1}\lambda\|\boldsymbol{\mathfrak{M}}(\boldsymbol{\mathcal{H}})\|_{2}-2(1-\eta_{2})\lambda\|\boldsymbol{\mathcal{H}}_{\Gamma^{c}}\|_{*}+4\lambda\|\boldsymbol{\mathcal{X}}_{\Gamma^{c}}\|_{*},
\end{eqnarray*}
where (a) follows from (\ref{Property 2}) and (b) is due to (\ref{Property 1}), (\ref{Property 3}). The assumption (\ref{Th1_condition}) implies that
\begin{equation*}
 1-\eta_{2}=1-\frac{\sqrt{n_{3}}\delta_{tr}^{\boldsymbol{\mathfrak{M}}}}{\sqrt{(1-(\delta_{tr}^{\boldsymbol{\mathfrak{M}}})^{2})(t-1)}}>1-\frac{\sqrt{n_{3}}\sqrt{(t-1)/(n_{3}^{2}+t-1)}}{\sqrt{\left(1-(t-1)/(n_{3}^{2}+t-1)\right)(t-1)}}=0,
\end{equation*}
and hence
\begin{equation*}
\|\boldsymbol{\mathfrak{M}}(\boldsymbol{\mathcal{H}})\|_{2}^{2}-2(\sqrt{r}\eta_{1}\lambda+\epsilon)\|\boldsymbol{\mathfrak{M}}(\boldsymbol{\mathcal{H}})\|_{2}-4\lambda\|\boldsymbol{\mathcal{X}}_{\Gamma^{c}}\|_{*}\leq0,
\end{equation*}
which implies that
\begin{eqnarray*}
\left(\|\boldsymbol{\mathfrak{M}}(\boldsymbol{\mathcal{H}})\|_{2}-(\sqrt{r}\eta_{1}\lambda+\epsilon)\right)^{2}&\leq&(\sqrt{r}\eta_{1}\lambda+\epsilon)^{2}+4\lambda\|\boldsymbol{\mathcal{X}}_{\Gamma^{c}}\|_{*}\\
&\leq&\left(\sqrt{r}\eta_{1}\lambda+\epsilon+\frac{2\lambda\|\boldsymbol{\mathcal{X}}_{\Gamma^{c}}\|_{*}}{\sqrt{r}\eta_{1}\lambda+\epsilon}\right)^{2}\\
&\leq&\left(\sqrt{r}\eta_{1}\lambda+\epsilon+\frac{2\|\boldsymbol{\mathcal{X}}_{\Gamma^{c}}\|_{*}}{\sqrt{r}\eta_{1}}\right)^{2}.
\end{eqnarray*}
Therefore, we conclude that (\ref{Theorem-1-results1}) holds. Plugging (\ref{Theorem-1-results1}) into (\ref{Lemma3-results2}), by $\|\boldsymbol{\mathcal{H}}_{\Gamma}\|_{*}\leq\sqrt{r}\|\boldsymbol{\mathcal{H}}_{\Gamma}\|_{F}$, we get
\begin{eqnarray}\label{The proof of Theorem-ineq1}
\|\boldsymbol{\mathcal{H}}_{\Gamma^{c}}\|_{*}&\leq&\|\boldsymbol{\mathcal{H}}_{\Gamma}\|_{*}+2\|\boldsymbol{\mathcal{X}}_{\Gamma^{c}}\|_{*}+\frac{\epsilon}{\lambda}\left(\frac{2\|\boldsymbol{\mathcal{X}}_{\Gamma^{c}}\|_{*}}{\sqrt{r}\eta_{1}}+2\sqrt{r}\eta_{1}\lambda+2\epsilon\right)\nonumber\\
&\leq&\sqrt{r}\|\boldsymbol{\mathcal{H}}_{\Gamma}\|_{F}+\frac{2(\sqrt{r}\eta_{1}\lambda+\epsilon)}{\sqrt{r}\eta_{1}\lambda}\|\boldsymbol{\mathcal{X}}_{\Gamma^{c}}\|_{*}+\frac{\epsilon}{\lambda}(2\sqrt{r}\eta_{1}\lambda+2\epsilon).
\end{eqnarray}
Combining (\ref{Lemma2-results}), (\ref{Theorem-1-results1}) and (\ref{The proof of Theorem-ineq1}) yields
\begin{eqnarray*}
 \|\boldsymbol{\mathcal{H}}_{\Gamma}\|_{F}&\leq&\eta_{1}\left(\frac{2\|\boldsymbol{\mathcal{X}}_{\Gamma^{c}}\|_{*}}{\sqrt{r}\eta_{1}}+2\sqrt{r}\eta_{1}\lambda+2\epsilon\right)\\
&&+\frac{\eta_{2}}{\sqrt{r}}\left(\sqrt{r}\|\boldsymbol{\mathcal{H}}_{\Gamma}\|_{F}+\frac{2(\sqrt{r}\eta_{1}\lambda+\epsilon)}{\sqrt{r}\eta_{1}\lambda}\|\boldsymbol{\mathcal{X}}_{\Gamma^{c}}\|_{*}+\frac{\epsilon}{\lambda}(2\sqrt{r}\eta_{1}\lambda+2\epsilon)\right)\\
&=&\eta_{2}\|\boldsymbol{\mathcal{H}}_{\Gamma}\|_{F}+\frac{2\sqrt{r}\eta_{1}(1+\eta_{2})\lambda+2\eta_{2}\epsilon}{r\eta_{1}\lambda}\|\boldsymbol{\mathcal{X}}_{\Gamma^{c}}\|_{*}+(\eta_{1}+\frac{\epsilon}{\lambda})(2\sqrt{r}\eta_{1}\lambda+2\epsilon).
\end{eqnarray*}
Note that $1-\eta_{2}>0$, so the above inequality leads to
\begin{equation}\label{The proof of Theorem-ineq2}
 \|\boldsymbol{\mathcal{H}}_{\Gamma}\|_{F}\leq\frac{2\sqrt{r}\eta_{1}(1+\eta_{2})\lambda+2\eta_{2}\epsilon}{r\eta_{1}(1-\eta_{2})\lambda}\|\boldsymbol{\mathcal{X}}_{\Gamma^{c}}\|_{*}+\frac{(\eta_{1}\lambda+\epsilon)(2\sqrt{r}\eta_{1}\lambda+2\epsilon)}{(1-\eta_{2})\lambda}.
\end{equation}
To prove (\ref{Theorem-1-results2}), application of (\ref{The proof of Theorem-ineq1}) and (\ref{The proof of Theorem-ineq2}) yields
\begin{eqnarray*}
\|\boldsymbol{\mathcal{H}}\|_{F}&\leq&\|\boldsymbol{\mathcal{H}}_{\Gamma}\|_{F}+\|\boldsymbol{\mathcal{H}}_{\Gamma^{c}}\|_{F}\\ &\leq&(\sqrt{n_{3}r}+1)\|\boldsymbol{\mathcal{H}}_{\Gamma}\|_{F}+\frac{2\sqrt{n_{3}}(\sqrt{r}\eta_{1}\lambda+\epsilon)}{\sqrt{r}\eta_{1}\lambda}\|\boldsymbol{\mathcal{X}}_{\Gamma^{c}}\|_{*}+\frac{\epsilon}{\lambda}\sqrt{n_{3}}(2\sqrt{r}\eta_{1}\lambda+2\epsilon)\\
&\leq&(\sqrt{n_{3}r}+1)\left(\frac{2\sqrt{r}\eta_{1}(1+\eta_{2})\lambda+2\eta_{2}\epsilon}{r\eta_{1}(1-\eta_{2})\lambda}\|\boldsymbol{\mathcal{X}}_{\Gamma^{c}}\|_{*}+\frac{(\eta_{1}\lambda+\epsilon)(2\sqrt{r}\eta_{1}\lambda+2\epsilon)}{(1-\eta_{2})\lambda}\right)\\
&&+\frac{2\sqrt{n_{3}}(\sqrt{r}\eta_{1}\lambda+\epsilon)}{\sqrt{r}\eta_{1}\lambda}\|\boldsymbol{\mathcal{X}}_{\Gamma^{c}}\|_{*}+\frac{\epsilon}{\lambda}\sqrt{n_{3}}(2\sqrt{r}\eta_{1}\lambda+2\epsilon)\\
&\leq&\frac{2\sqrt{r}\eta_{1}(2\sqrt{n_{3}r}+1+\eta_{2})\lambda+2(\sqrt{n_{3}r}+\eta_{2})\epsilon}{r\eta_{1}(1-\eta_{2})\lambda}\|\boldsymbol{\mathcal{X}}_{\Gamma^{c}}\|_{*}\\
&+&\frac{(\sqrt{n_{3}r}+1)\eta_{1}\lambda+(\sqrt{n_{3}r}-\sqrt{n_{3}}\eta_{2}+\sqrt{n_{3}}+1)\epsilon}{(1-\eta_{2})\lambda(2\sqrt{r}\eta_{1}\lambda+2\epsilon)^{-1}},
\end{eqnarray*}
where the second inequality is because of $\|\boldsymbol{\mathcal{H}}_{\Gamma^{c}}\|_{F}=\frac{1}{\sqrt{n_{3}}}\|\bdiag(\boldsymbol{\bar{\mathcal{H}}}_{\Gamma^{c}})\|_{F}\leq\frac{1}{\sqrt{n_{3}}}\|\bdiag(\boldsymbol{\bar{\mathcal{H}}}_{\Gamma^{c}})\|_{*}=\sqrt{n_{3}}\|\boldsymbol{\mathcal{H}}_{\Gamma^{c}}\|_{*}$. So far, we have completed the proof.
\end{proof}

We note that the obtained t-RIC condition (\ref{Th1_condition}) is related to the length $n_{3}$ of the third dimension. This is due to the fact that the discrete Fourier transform (DFT) is performed along the third dimension of $\boldsymbol{\mathcal{X}}$. Further, we want to stress that this crucial quantity $n_{3}$ is rigorously deduced from the t-product and makes the result of the tensor consistent with the matrix case. For general problems, let $n_{3}$ be the smallest size of three modes of the third-order tensor, e.g. $n_{3}=3$ for the third-order tensor $\boldsymbol{\mathcal{X}}\in\mathbb{R}^{h\times w \times 3}$ from a color image with size $h\times w$, where three frontal slices correspond to the R, G, B channels; $n_{3}=8$ for 3-D face detection using tensor data $\boldsymbol{\mathcal{X}}\in\mathbb{R}^{h\times w \times 8}$ with column $h$, row $w$, and depth mode $8$. Especially, when $n_{3}=1$, our model (\ref{tensor nuclear norm min unconstrained}) returns to the case of LRMR and the condition (\ref{Th1_condition}) degenerates to $\delta_{tr}^{\boldsymbol{\mathcal{M}}}<\sqrt{(t-1)/t}$ which has also been proved to be sharp by Cai, et al. \cite{Cai2013Sharp} for stable recovery via the constrained nuclear norm minimization for $t>4/3$. We note that, to the best of our knowledge, results like our Theorem \ref{Theorem-1} has not previously been reported in the literature.

Theorem \ref{Theorem-1} not only offers a sufficient condition for stably recovering tensor $\boldsymbol{\mathcal{X}}$ based on solving (\ref{tensor nuclear norm min unconstrained}), but also provides an error upper bound estimate for the recovery of tensor $\boldsymbol{\mathcal{X}}$ via RTNNM model.
This result clearly depicts the relationship among reconstruction error, the best $r$-term approximation, noise level $\epsilon$ and $\lambda$. There exist some special cases of Theorem \ref{Theorem-1} which is worth studying. For examples, one can associate the $\ell_2$-norm bounded noise level $\epsilon$ with the trade-off parameter $\lambda$ (such as $\epsilon=\lambda/2$) as \cite{Ge2018Stable,Shen2015Stable,Candes2011Tight}. This case can be summarized by Corollary \ref{corollary}. Notice that we can take a $\lambda$ which is close to zero such that $\tilde{C_{2}}\lambda$ and $\tilde{C_{4}}\lambda$ in (\ref{corollary-results1}),(\ref{corollary-results2}) are close to zero for the noise-free case $\boldsymbol{w}=\boldsymbol{0}$. Then Corollary \ref{corollary} shows that tensor $\boldsymbol{\mathcal{X}}$ can be approximately recovery by solving (\ref{tensor nuclear norm min unconstrained}) if $\|\boldsymbol{\mathcal{X}}_{-\max(r)}\|_{*}$ is small.

\begin{corollary}\label{corollary}
Suppose that the noise measurements $\boldsymbol{y}=\boldsymbol{\mathfrak{M}}(\boldsymbol{\mathcal{X}})+\boldsymbol{w}$ of tensor $\boldsymbol{\mathcal{X}}\in\mathbb{R}^{n_{1}\times n_{2}\times n_{3}}$ are observed with noise level $\|\boldsymbol{w}\|_{2}\leq\epsilon=\frac{\lambda}{2}$. If $\boldsymbol{\mathfrak{M}}$ satisfies t-RIP with
\begin{equation}\label{corollary_condition}
\delta_{tr}^{\boldsymbol{\mathfrak{M}}}<\sqrt{\frac{t-1}{n_{3}^{2}+t-1}}
\end{equation}
for certain $t>1$, then we have
\begin{equation}
\|\boldsymbol{\mathfrak{M}}(\boldsymbol{\hat{\mathcal{X}}}-\boldsymbol{\mathcal{X}})\|_{2}\leq \tilde{C_{1}}\|\boldsymbol{\mathcal{X}}_{-\max(r)}\|_{*}+ \tilde{C_{2}}\lambda,\label{corollary-results1}
\end{equation}
and
\begin{equation}\label{corollary-results2}
\|\boldsymbol{\hat{\mathcal{X}}}-\boldsymbol{\mathcal{X}}\|_{F}\leq \tilde{C_{3}}\|\boldsymbol{\mathcal{X}}_{-\max(r)}\|_{*}+ \tilde{C_{4}}\lambda,
\end{equation}
where $\boldsymbol{\hat{\mathcal{X}}}$ is the solution to (\ref{tensor nuclear norm min unconstrained}), and $\tilde{C_{i}}, i=1,2,3,4$ are denoted as
\begin{eqnarray*}
\tilde{C_{1}}&=&\frac{2}{\sqrt{r}\eta_{1}},\quad\tilde{C_{2}}=2\sqrt{r}\eta_{1}+1,\\
\tilde{C_{3}}&=&\frac{2\sqrt{r}\eta_{1}(2\sqrt{n_{3}r}+1+\eta_{2})+\sqrt{n_{3}r}+\eta_{2}}{r\eta_{1}(1-\eta_{2})},\\
\tilde{C_{4}}&=&\frac{2(\sqrt{n_{3}r}+1)\eta_{1}+\sqrt{n_{3}r}-\sqrt{n_{3}}\eta_{2}+\sqrt{n_{3}}+1}{2(1-\eta_{2})(2\sqrt{r}\eta_{1}+1)^{-1}}.
\end{eqnarray*}
\end{corollary}

\section{Numerical experiments}
\label{Numerical experiments}
In this section, we present several numerical experiments to corroborate our analysis.
\subsection{Optimization by ADMM}
We perform $\boldsymbol{y}=\boldsymbol{M}\rm vec(\boldsymbol{\mathcal{X}})+\boldsymbol{w}$ to get the linear noise measurements instead of $\boldsymbol{y}=\boldsymbol{\mathfrak{M}}(\boldsymbol{\mathcal{X}})+\boldsymbol{w}$. Then the RTNNM model \eqref{tensor nuclear norm min unconstrained} can be reformulated as
\begin{equation}\label{reshape RTNNM}
  \min \limits_{{\boldsymbol{\mathcal{X}}}\in\mathbb{R}^{n_{1} \times n_{2} \times n_{3}}}~\|\boldsymbol{\mathcal{X}}\|_{*}+\frac{1}{2\lambda}\|\boldsymbol{y}-\boldsymbol{M}\rm vec(\boldsymbol{\mathcal{X}})\|_{2}^{2},
\end{equation}
where $\boldsymbol{y},\boldsymbol{w}\in\mathbb{R}^{m}$, $\boldsymbol{\mathcal{X}}\in\mathbb{R}^{n_{1} \times n_{2} \times n_{3}}$, $\boldsymbol{M}\in\mathbb{R}^{m\times (n_{1}n_{2}n_{3})}$ is a Gaussian measurement ensemble and $\rm vec(\boldsymbol{\mathcal{X}})$ denotes the vectorization of $\boldsymbol{\mathcal{X}}$. We adopt the alternating direction method of multipliers (ADMM) \cite{boyd2011distributed} to solve this kind of problem quickly and accurately. We firstly introduce an auxiliary variable $\boldsymbol{\mathcal{Z}}\in\mathbb{R}^{n_{1} \times n_{2} \times n_{3}}$ so that \eqref{reshape RTNNM} forms a constrained optimization problem
\begin{equation*}\label{constrained RTNNM}
  \min \limits_{{\boldsymbol{\mathcal{X}}}\in\mathbb{R}^{n_{1} \times n_{2} \times n_{3}}}~\|\boldsymbol{\mathcal{X}}\|_{*}+\frac{1}{2\lambda}\|\boldsymbol{y}-\boldsymbol{M}\rm vec(\boldsymbol{\mathcal{Z}})\|_{2}^{2},~~s.t.~~\boldsymbol{\mathcal{X}}=\boldsymbol{\mathcal{Z}}.
\end{equation*}
The augmented Lagrangian function of the above constrained optimization problem is
\begin{align*}\label{the augmented Lagrangian function}
L(\boldsymbol{\mathcal{X}},\boldsymbol{\mathcal{Z}},\boldsymbol{\mathcal{K}}) = \lambda\|\boldsymbol{\mathcal{X}}\|_{*}+\frac{1}{2}\|\boldsymbol{y}-\boldsymbol{M}\rm vec(\boldsymbol{\mathcal{Z}})\|_{2}^{2}+\langle\boldsymbol{\mathcal{K}},\boldsymbol{\mathcal{X}}-\boldsymbol{\mathcal{Z}}\rangle+\frac{\rho}{2}\|\boldsymbol{\mathcal{X}}-\boldsymbol{\mathcal{Z}}\|_{2}^{2}
\end{align*}
where $\rho$ is a positive scalar and $\boldsymbol{\mathcal{K}}$ is the Lagrangian multiplier tensor. By minimizing the augmented Lagrangian function, we can obtain the closed-form solutions of the variables $\boldsymbol{\mathcal{X}}$ and $\boldsymbol{\mathcal{Z}}$. A detailed update process is shown in Algorithm \ref{our Algorithm}. In particular, according to Theorem 4.2 in \cite{Lu2019Tensor}, the proximal operator in Step 3 of Algorithm \ref{our Algorithm} can be computed by exploiting the tensor Singular Value Thresholding (t-SVT) algorithm.

\subsection{Experiment Results}
All numerical experiments are tested on a PC with 4 GB of RAM and Intel core i5-4200M (2.5GHz). In order to avoid randomness, we perform 50 times against each test and report the average result.
\begin{algorithm}[h]%the beginning of the algorithm
\caption{Algorithm for solving RTNNM \eqref{tensor nuclear norm min unconstrained}}\label{our Algorithm}
\begin{algorithmic}[1]%[1] will show the row number by line; [] will hide the number
\REQUIRE $\boldsymbol{M}\in\mathbb{R}^{m\times (n_{1}n_{2}n_{3})}$, $\boldsymbol{y}\in\mathbb{R}^{m}$.
\STATE Initialize $\boldsymbol{\mathcal{X}}_{0} = \boldsymbol{\mathcal{Z}}_{0} = \boldsymbol{\mathcal{K}}_{0} = \boldsymbol{0}$, $\rho_{0}=10^{-4}$, $\rho_{max}=10^{10}$, $\vartheta=1.5$, $\varpi=10^{-8}$ and $k=0$.
\WHILE{no convergence}
\STATE Update $\boldsymbol{\mathcal{X}}_{k+1}$ by $\boldsymbol{\mathcal{X}}_{k+1}=\arg\min\limits_{{\boldsymbol{\mathcal{X}}}} \lambda\|\boldsymbol{\mathcal{X}}\|_{*} + \frac{\rho_{k}}{2}\|\boldsymbol{\mathcal{X}}-\boldsymbol{\mathcal{Z}}_{k}+\frac{\boldsymbol{\mathcal{K}}_{k}}{\rho_{k}}\|_{F}^{2}$.
\STATE Update $\boldsymbol{\mathcal{Z}}_{k+1}$ by \\
$\boldsymbol{z}=\arg\min\limits_{{\boldsymbol{\mathcal{Z}}}} (\boldsymbol{M}^{T}\boldsymbol{M}+\rho_{k}\boldsymbol{I})^{-1}(\boldsymbol{M}^{T}\boldsymbol{y}+\rm vec(\boldsymbol{\mathcal{K}}_{k})+\rho_{k}\rm vec\left(\boldsymbol{\mathcal{X}}_{k+1})\right)$ and $\boldsymbol{\mathcal{Z}}_{k+1}\leftarrow \boldsymbol{z}:$ reshape $\boldsymbol{z}$ to the tensor $\boldsymbol{\mathcal{Z}}_{k+1}$ of size $n_{1} \times n_{2} \times n_{3}$.
\STATE Update $\boldsymbol{\mathcal{K}}_{k+1}$ by $\boldsymbol{\mathcal{K}}_{k+1}=\boldsymbol{\mathcal{K}}_{k}+\rho_{k}(\boldsymbol{\mathcal{X}}_{k+1}-\boldsymbol{\mathcal{Z}}_{k+1})$.
\STATE Update $\rho_{k+1}$ by $\rho_{k+1}=\min(\vartheta\rho_{k}, \rho_{max})$.
\STATE Check the convergence conditions\\
$\|\boldsymbol{\mathcal{X}}_{k+1}-\boldsymbol{\mathcal{X}}_{k}\|_{\infty}\leq\varpi$, $\|\boldsymbol{\mathcal{Z}}_{k+1}-\boldsymbol{\mathcal{Z}}_{k}\|_{\infty}\leq\varpi$, $\|\boldsymbol{\mathcal{X}}_{k+1}-\boldsymbol{\mathcal{Z}}_{k+1}\|_{\infty}\leq\varpi$.
\STATE Update $k\leftarrow k+1$.
\ENDWHILE
\ENSURE $\boldsymbol{\mathcal{X}}=\boldsymbol{\mathcal{X}}_{k}$, $\boldsymbol{\mathcal{Z}}=\boldsymbol{\mathcal{Z}}_{k}$ and $\boldsymbol{\mathcal{K}}=\boldsymbol{\mathcal{K}}_{k}$.
\end{algorithmic}
\end{algorithm}

First, we generate a tubal rank $r$ tensor $\boldsymbol{\mathcal{X}}\in\mathbb{R}^{n \times n \times n_{3}}$ as a product $\boldsymbol{\mathcal{X}}=\boldsymbol{\mathcal{X}}_{1} \ast \boldsymbol{\mathcal{X}}_{2}$ where $\boldsymbol{\mathcal{X}}_{1}\in\mathbb{R}^{n \times r \times n_{3}}$ and $\boldsymbol{\mathcal{X}}_{2}\in\mathbb{R}^{r \times n \times n_{3}}$ are two tensors with entries independently sampled from a standard Gaussian distribution. Next, we generate a measurement matrix $\boldsymbol{M}\in\mathbb{R}^{m\times (n^{2}n_{3})}$ with i.i.d. $\mathcal{N}(0,1/m)$ entries. Using $\boldsymbol{\mathcal{X}}$ and $\boldsymbol{M}$, the measurements $\boldsymbol{y}$ are produced by $\boldsymbol{y}=\boldsymbol{M}\rm vec(\boldsymbol{\mathcal{X}})+\boldsymbol{w}$, where $\boldsymbol{w}$ is the Gaussian white noise with mean $0$ and variance $\sigma^{2}$ (the greater $\sigma$, the greater the noise level $\epsilon$). We uniformly evaluate the recovery performance of the model by the signal-to-noise ratio (SNR) defined as $20\log(\|\boldsymbol{\mathcal{X}}\|_{F}/\|\boldsymbol{\mathcal{X}}-\hat{\boldsymbol{\mathcal{X}}}\|_{F})$ in decibels (dB) (the greater the SNR, the smaller the reconstruction error). The key to studying the RTNNM model \eqref{tensor nuclear norm min unconstrained} is to explain the relationship among reconstruction error, noise level $\epsilon$ and $\lambda$. Therefore, we design three sets of experiments to explain it. Case 1: $n=10$, $n_{3}=5$, $r=0.1n$; Case 2: $n=20$, $n_{3}=5$, $r=0.2n$; Case 3: $n=30$, $n_{3}=5$, $r=0.3n$. According to \cite{zhang2019tensor}, the number of samples in case 1 is set to $2r(2n+1)n_{3}$, and the number of samples in case 2 and 3 is set to $1.5r(2n+1)n_{3}$.

All SNR values for $\sigma$ varying among $\{0.01,0.03,0.05,0.07,0.1\}$ and regularization parameters $\lambda$ varying among $\{10^{1},10^{0},10^{-1},10^{-2},10^{-3},10^{-4}\}$ in three cases are provided in Table \ref{tab1} with the best results highlighted in bold. It can be seen that there exist a consistent phenomenon for low-tubal-rank tensor recovery at different scales. When the regularization parameter $\lambda=10^{-1}$, as the standard deviation $\sigma$ increases, the SNR gradually decreases. There is only a slight trend for other regularization parameters $\lambda$. In addition, for each fixed noise level, the regularization parameter $\lambda=10^{-1}$ corresponds to the maximum SNR, which means the best of the low-tubal-rank tensor recovery. Therefore, $\lambda=10^{-1}$ is the optimal regularization parameter of the RTNNM model \eqref{tensor nuclear norm min unconstrained} in three cases. We plot the data in Table \ref{tab1} as Figure \ref{SNR}, which allows us to see the results of the above analysis at a glance. Thus, these experiments clearly demonstrate the quantitative correlation among reconstruction error, noise level $\epsilon$ and $\lambda$.

\begin{table*}[!htb]
  \centering
  \fontsize{10}{8}\selectfont
  \begin{threeparttable}
  \caption{SNR for different noise levels and regularization parameters.}\label{Tabel_Noise case 1}
    \begin{tabular}{ccccccc}
    \multicolumn{6}{c}{Case 1: $n=10$, $n_{3}=5$, $r=0.1n$}\cr
    \toprule
    \multirow{1}{*}{SNR (dB)}&

    $\sigma_{1}=0.01$&$\sigma_{2}=0.03$&$\sigma_{3}=0.05$&$\sigma_{4}=0.07$&$\sigma_{5}=0.1$\cr
    \midrule
    $\lambda_{1}=10^{1}$&      12.0980&12.0981&12.0970&12.0948&12.0892\cr
    $\lambda_{2}=10^{0}$&      25.8704&25.7607&25.5490&25.2479&24.6683\cr
    $\lambda_{3}=10^{-1}$&     \textbf{31.0711}&\textbf{30.4264}&\textbf{29.4052}&\textbf{28.2704}&\textbf{26.6232}\cr
    $\lambda_{4}=10^{-2}$&     22.7064&22.5787&22.3462&22.0335&21.4687\cr
    $\lambda_{5}=10^{-3}$&     5.5688&5.5679&5.5660&5.5632&5.5573\cr
    $\lambda_{6}=10^{-4}$&     2.7720&2.7718&2.7713&2.7706&2.7691\cr
    \bottomrule
    \multicolumn{6}{c}{Case 2: $n=20$, $n_{3}=5$, $r=0.2n$}&\cr
    \toprule
    \multirow{1}{*}{SNR (dB)}&

    $\sigma_{1}=0.01$&$\sigma_{2}=0.03$&$\sigma_{3}=0.05$&$\sigma_{4}=0.07$&$\sigma_{5}=0.1$\cr
    \midrule
    $\lambda_{1}=10^{1}$&      11.0030&11.0015&10.9986&10.9941&10.9845\cr
    $\lambda_{2}=10^{0}$&      20.6674&20.6139&20.5124&20.3671&20.0778\cr
    $\lambda_{3}=10^{-1}$&     \textbf{23.6323}&\textbf{23.4598}&\textbf{23.1416}&\textbf{22.7156}&\textbf{21.9574}\cr
    $\lambda_{4}=10^{-2}$&     19.5214&19.4528&19.3200&19.1312&18.7651\cr
    $\lambda_{5}=10^{-3}$&     6.4141&6.4128&6.4102&6.4065&6.3984\cr
    $\lambda_{6}=10^{-4}$&     4.5625&4.5620&4.5608&4.5591&4.5555\cr
    \bottomrule
    \multicolumn{6}{c}{Case 3: $n=30$, $n_{3}=5$, $r=0.3n$}&\cr
    \toprule
    \multirow{1}{*}{SNR (dB)}&

    $\sigma_{1}=0.01$&$\sigma_{2}=0.03$&$\sigma_{3}=0.05$&$\sigma_{4}=0.07$&$\sigma_{5}=0.1$\cr
    \midrule
    $\lambda_{1}=10^{1}$&      14.7061&14.7034&14.6986 &14.6917&14.6774 \cr
    $\lambda_{2}=10^{0}$&      30.0564&29.8961&29.5879 &29.1587&28.3575 \cr
    $\lambda_{3}=10^{-1}$&     48.0263&\textbf{41.8665}&\textbf{37.3142}&\textbf{33.9986}&\textbf{30.3401} \cr
    $\lambda_{4}=10^{-2}$&     \textbf{50.1381}&38.8254&33.6326&30.2868&26.8153 \cr
    $\lambda_{5}=10^{-3}$&     14.1284&14.1003&14.0445 &13.9624&13.7947 \cr
    $\lambda_{6}=10^{-4}$&     11.0540&11.0440&11.0242 &10.9946&10.9324 \cr
    \bottomrule
    \end{tabular}
    \label{tab1}
    \end{threeparttable}
\end{table*}

\begin{figure}[!htbp]
\begin{center}
\subfigure[]{\includegraphics[width=1\textwidth]{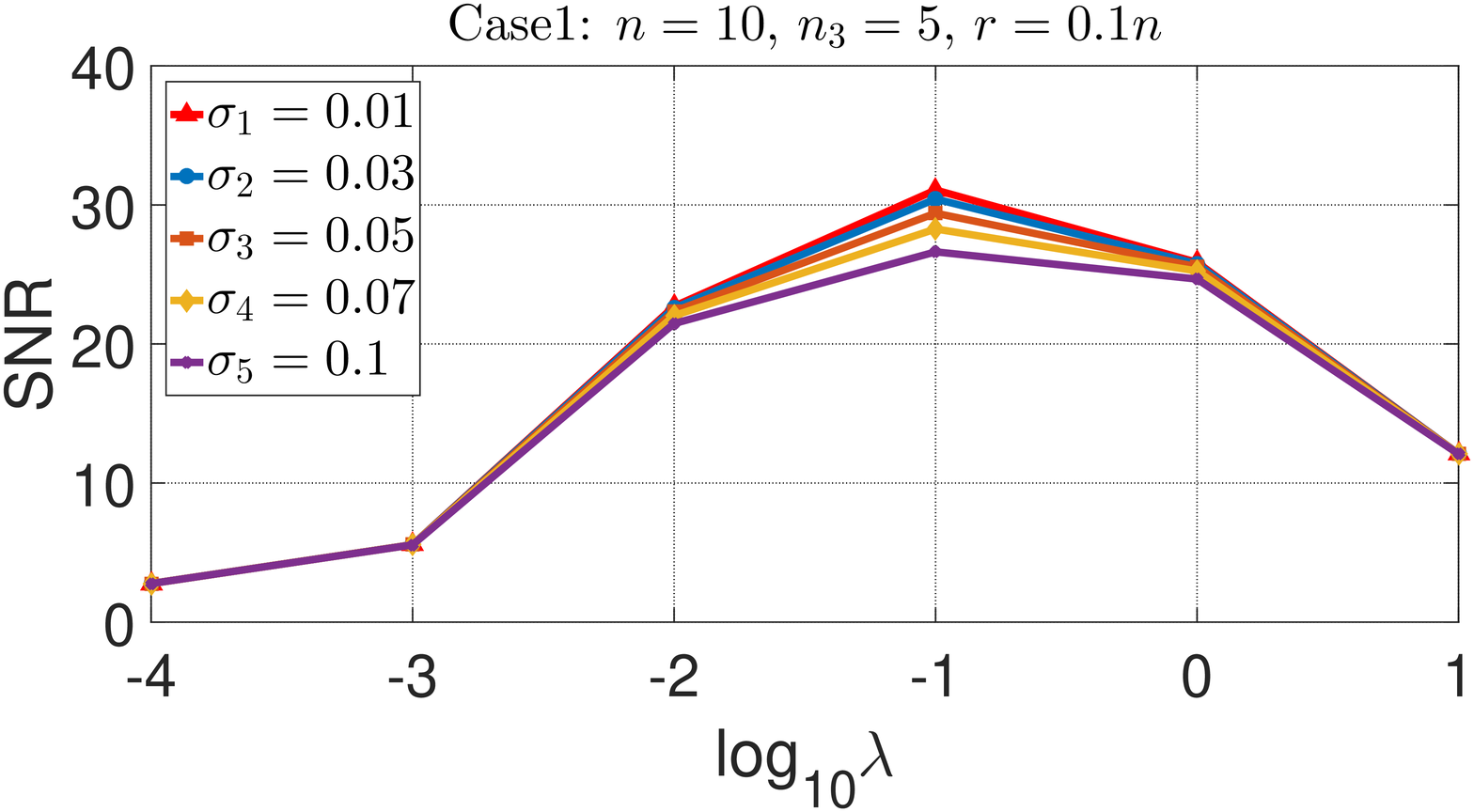}}
\hspace*{9pt}
\subfigure[]{\includegraphics[width=1\textwidth]{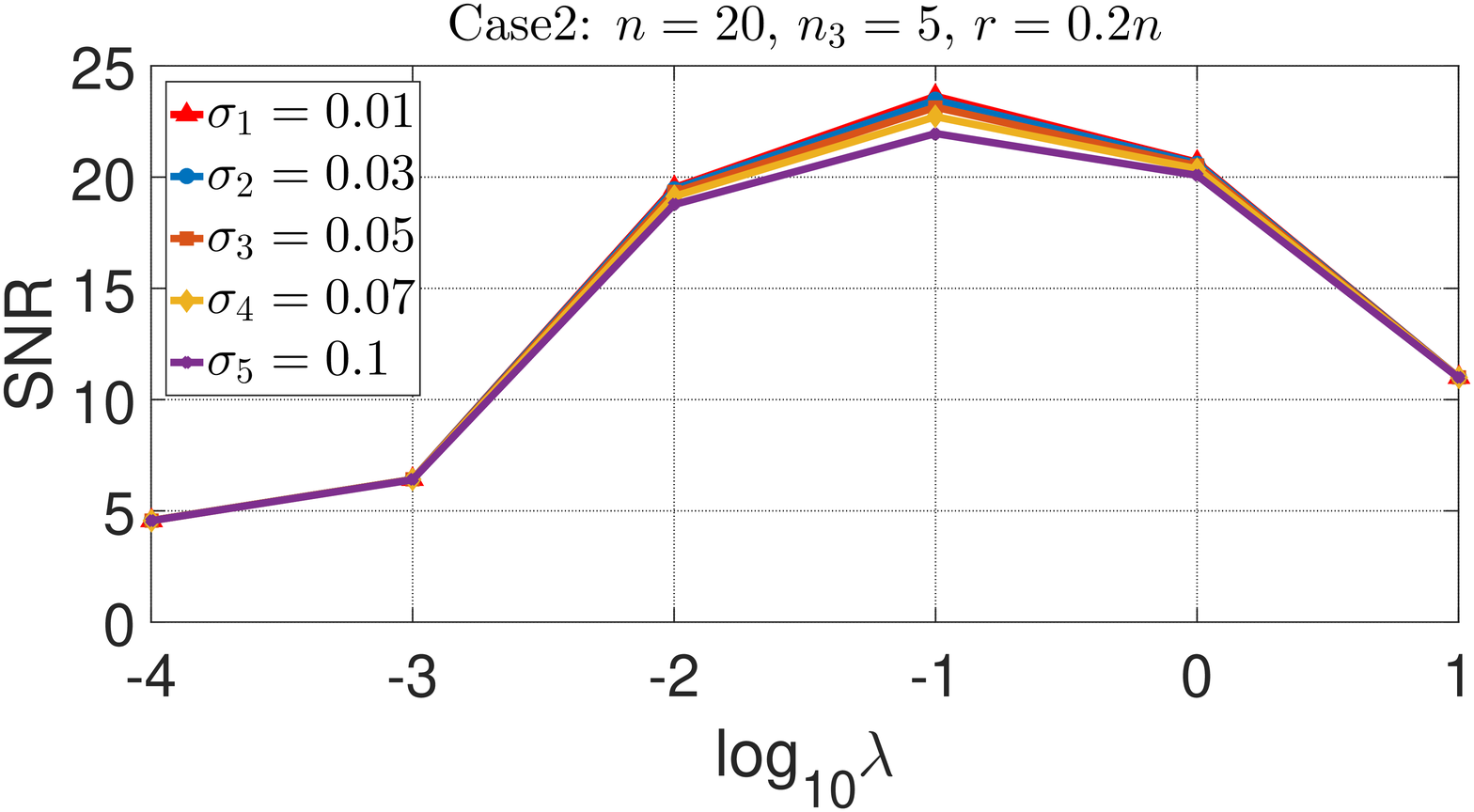}}
\hspace*{9pt}
\subfigure[]{\includegraphics[width=1\textwidth]{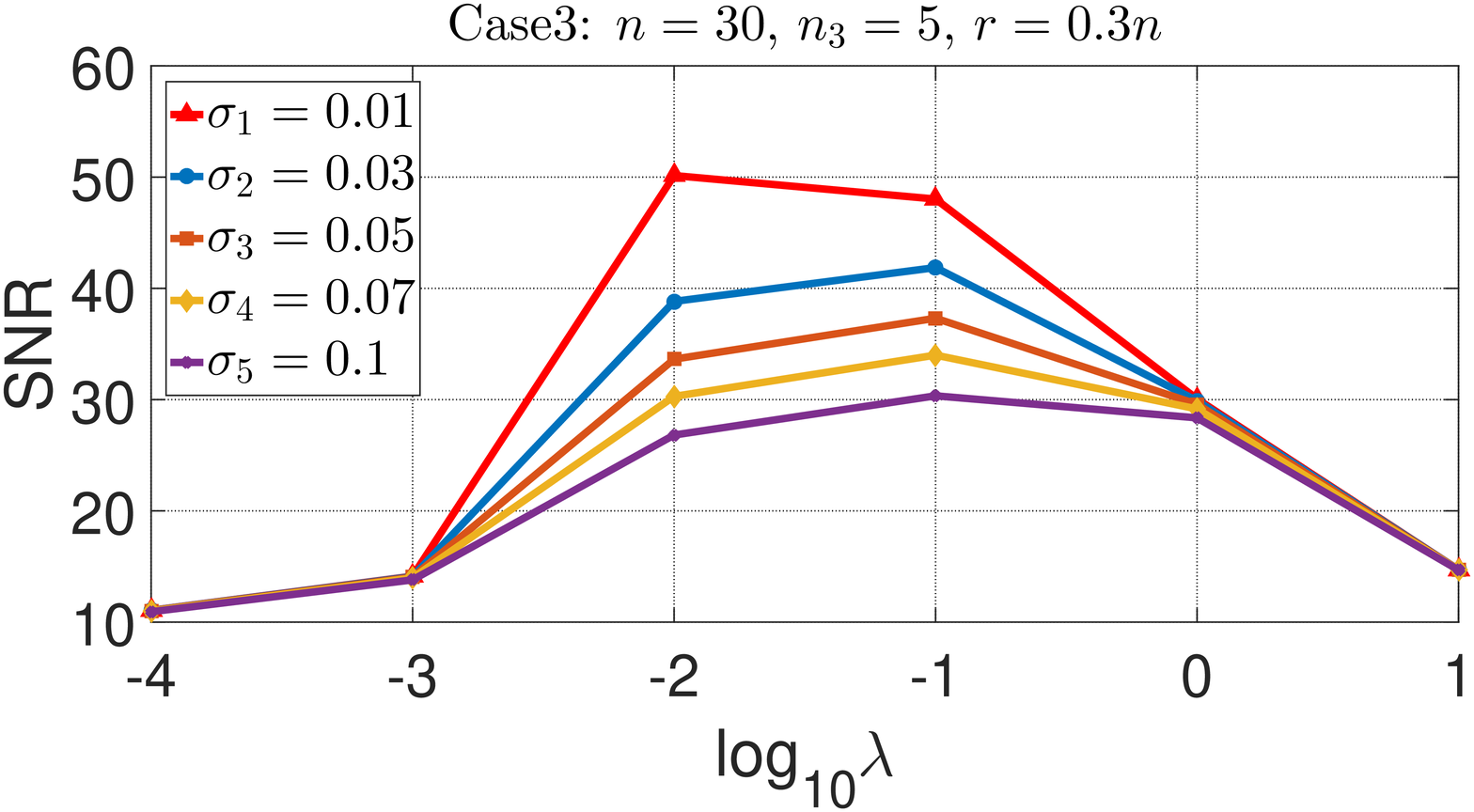}}
\end{center}
\vspace*{-14pt}
\caption{SNR for different noise levels and regularization parameters. (a) $\lambda$ versus SNR with $n=10$, $n_{3}=5$ and $r=0.1n$. (b) $\lambda$ versus SNR with $n=20$, $n_{3}=5$ and $r=0.2n$. (c) $\lambda$ versus SNR with $n=30$, $n_{3}=5$ and $r=0.3n$.}\label{SNR}
\end{figure}

\subsection{Discussion}
According to the established theoretical results in Theorem \ref{Theorem-1}, to obtain an exact low tubal rank solution (i.e., the case when $\epsilon=0$ and $\boldsymbol{\mathcal{X}}$ is a tensor with tubal rank $r$) of RTNNM model \eqref{tensor nuclear norm min unconstrained}, one has to set $\lambda$ to be $0$. However, when it comes to the practice use, such a setting is obviously impossible. In fact, the practical wisdom (see, e.g., \cite{lai2013improved,yin2015minimization,hu2017group}) also indicates that the selection of $\lambda$ may also be effected by the data scale, sampling number, approximation accuracy, objective function, low-rankness of the desired tensor, etc. On the other hand, if taking a closer look at the obtained theoretical results in the presence of noise (i.e., $\epsilon\neq0$), one will find that setting the lambda to be a sufficient big or small value is also not helpful to approach an expected approximate solution. In a word, one has to carefully select such a regularization parameter to obtained a desired optimal solution. This is exactly what we are trying to explain in our experiments.

\section{Conclusion}
\label{Conclusion}
In this paper, a heuristic notion of tensor Restricted Isometry Property (t-RIP) has been introduced based on tensor Singular Value Decomposition (t-SVD). Comparing with other definitions \cite{shi2013guarantees,rauhut2017low}, it is more representative as a higher-order generalization of the traditional RIP for vector and matrix recovery since the forms and properties of t-RIP and t-SVD are consistent with the vector/matrix case. This point is crucial because this guarantees that our theoretical investigation can be done in a similar way as CS and LRMR. A sufficient condition was presented, based on the RTNNM model, for stably recovering a given low-tubal-rank tensor that is corrupted with an $\ell_{2}$-norm bounded noise. However, this condition only considers the $\delta_{tr}^{\boldsymbol{\mathfrak{M}}}$ of the map $\boldsymbol{\mathfrak{M}}$ when $t$ is limited to $t>1$. In the future, we hope to provide a complete answer for $\delta_{tr}^{\boldsymbol{\mathfrak{M}}}$ when $0 < t \leq 1$.

\section*{Acknowledgment}
This work was supported by National Natural Science Foundation of China (Grant Nos. 61673015, 61273020), Fundamental Research Funds for the Central Universities (Grant Nos. XDJK2018C076, SWU1809002) and China Postdoctoral Science Foundation (Grant No. 2018M643390) and Graduate Student Scientific Research Innovation Projects in Chongqing (Grant No. CYB19083).

\appendix
\section{Proof of Lemma \ref{Lemma-2}}\label{Appendix A}
\begin{proof}
\emph{\textbf{Step 1: Sparse Representation of a Polytope}}.

Without loss of generality, assume that $tr$ is an integer for a given $t>1$. Next we divide the index set $\Gamma^{c}$ into two disjoint subsets, that is,
\begin{equation*}
\Gamma_{1}=\{i\in \Gamma^{c}: \boldsymbol{\mathcal{S}}_{\boldsymbol{\mathcal{H}}}(i,i,1)>\phi\},~~\Gamma_{2}=\{i\in \Gamma^{c}: \boldsymbol{\mathcal{S}}_{\boldsymbol{\mathcal{H}}}(i,i,1)\leq\phi\},
\end{equation*}
where $\phi\triangleq\|\boldsymbol{\mathcal{H}}_{\Gamma^{c}}\|_{*}/((t-1)r)$. Clearly,
\begin{equation*}
\Gamma_{1}\cup \Gamma_{2}=\Gamma^{c}~\rm and~\Gamma_{1}\cap \Gamma=\emptyset,
\end{equation*}
which implies that $\boldsymbol{\mathcal{H}}=\boldsymbol{\mathcal{H}}_{\Gamma}+\boldsymbol{\mathcal{H}}_{\Gamma^{c}}=\boldsymbol{\mathcal{H}}_{\Gamma}+\boldsymbol{\mathcal{H}}_{\Gamma_{1}}+\boldsymbol{\mathcal{H}}_{\Gamma_{2}}$ and $\|\boldsymbol{\mathcal{H}}_{\Gamma}\|_{F}\leq\|\boldsymbol{\mathcal{H}}_{\Gamma\cup \Gamma_{1}}\|_{F}$, respectively. In order to prove (\ref{Lemma2-results}), we only need to check
\begin{equation}\label{Lemma2-proof-inequality1}
\|\boldsymbol{\mathcal{H}}_{\Gamma\cup \Gamma_{1}}\|_{F}\leq \eta_{1}\|\boldsymbol{\mathfrak{M}}(\boldsymbol{\mathcal{H}})\|_{2}+\frac{\eta_{2}}{\sqrt{r}}\|\boldsymbol{\mathcal{H}}_{\Gamma^{c}}\|_{*}.
\end{equation}
%By denoting $|\Gamma_{1}|=s$, we first show that
%\begin{align}\label{Lemma2-proof-inequality2}
%s<(t-1)r.
%\end{align}
%If $\Gamma_{1}=\emptyset$, then $s=0$ and hence \eqref{Lemma2-proof-inequality2} holds.
Let $\|\boldsymbol{s}_{\boldsymbol{\mathcal{H}}_{\Gamma_{1}}}\|_{1}\triangleq\sum_{i\in\Gamma_{1}}\boldsymbol{\mathcal{S}}_{\boldsymbol{\mathcal{H}}}(i,i,1)=\|\boldsymbol{\mathcal{H}}_{\Gamma_{1}}\|_{*}$, where $\boldsymbol{s}_{\boldsymbol{\mathcal{H}}_{\Gamma_{1}}}$ is denoted as the diagonal vector of first frontal slice of $\boldsymbol{\mathcal{S}}_{\boldsymbol{\mathcal{H}}}$ whose element $\boldsymbol{\mathcal{S}}_{\boldsymbol{\mathcal{H}}_{\Gamma_{1}}}(i,i,1)=\boldsymbol{\mathcal{S}}_{\boldsymbol{\mathcal{H}}}(i,i,1)$ for $i\in\Gamma_{1}$ and $\boldsymbol{\mathcal{S}}_{\boldsymbol{\mathcal{H}}_{\Gamma_{1}}}(i,i,1)=0$ otherwise. Since all non-zero entries of vector $\boldsymbol{s}_{\boldsymbol{\mathcal{H}}_{\Gamma_{1}}}$ have magnitude larger than $\phi$, we have,
\begin{equation*}
\|\boldsymbol{s}_{\boldsymbol{\mathcal{H}}_{\Gamma_{1}}}\|_{1}=\|\boldsymbol{\mathcal{H}}_{\Gamma_{1}}\|_{*}>|\Gamma_{1}|\frac{\|\boldsymbol{\mathcal{H}}_{\Gamma^{c}}\|_{*}}{(t-1)r}
\geq|\Gamma_{1}|\frac{\|\boldsymbol{\mathcal{H}}_{\Gamma_{1}}\|_{*}}{(t-1)r}=\frac{|\Gamma_{1}|}{(t-1)r}\|\boldsymbol{s}_{\boldsymbol{\mathcal{H}}_{\Gamma_{1}}}\|_{1}.
\end{equation*}
Namely $|\Gamma_{1}|<(t-1)r$. Besides, we also have
\begin{equation*}
\|\boldsymbol{s}_{\boldsymbol{\mathcal{H}}_{\Gamma_{2}}}\|_{1}=\|\boldsymbol{\mathcal{H}}_{\Gamma_{2}}\|_{*}=\|\boldsymbol{\mathcal{H}}_{\Gamma^{c}}\|_{*}-\|\boldsymbol{\mathcal{H}}_{\Gamma_{1}}\|_{*}
\leq((t-1)r-|\Gamma_{1}|)\phi
\end{equation*}
and
\begin{equation*}
\|\boldsymbol{s}_{\boldsymbol{\mathcal{H}}_{\Gamma_{2}}}\|_{\infty}\triangleq\max_{i\in \Gamma_{2}}\boldsymbol{\mathcal{S}}_{\boldsymbol{\mathcal{H}}}(i,i,1)\leq\phi.
\end{equation*}
Now, since $\boldsymbol{s}_{\boldsymbol{\mathcal{H}}_{\Gamma_{2}}}\in T(\phi,(t-1)r-|\Gamma_{1}|)$, applying Lemma \ref{Sparse-Representation}, $\boldsymbol{s}_{\boldsymbol{\mathcal{H}}_{\Gamma_{2}}}$ can be rewritten as:
\begin{equation*}
  \boldsymbol{s}_{\boldsymbol{\mathcal{H}}_{\Gamma_{2}}}=\sum_{i=1}^{N}\gamma_{i}\boldsymbol{g}_{i},
\end{equation*}
where $\boldsymbol{g}_{i}\in U(\phi, (t-1)r-|\Gamma_{1}|, \boldsymbol{s}_{\boldsymbol{\mathcal{H}}_{\Gamma_{2}}})$ and $0\leq\gamma_{i}\leq1$, $\sum_{i=1}^{N}\gamma_{i}=1$.\\

\emph{\textbf{Step 2: Consequence of t-RIP}}.

Furthermore, define
\begin{eqnarray*}
\boldsymbol{s}_{\boldsymbol{\mathcal{B}}_{i}}&=&(1+\delta_{tr}^{\boldsymbol{\mathfrak{M}}})\boldsymbol{s}_{\boldsymbol{\mathcal{H}}_{\Gamma\cup\Gamma_{1}}}+\delta_{tr}^{\boldsymbol{\mathfrak{M}}}\boldsymbol{s}_{\boldsymbol{\mathcal{G}}_{i}},\quad\boldsymbol{s}_{\boldsymbol{\mathcal{P}}_{i}}=(1-\delta_{tr}^{\boldsymbol{\mathfrak{M}}})\boldsymbol{s}_{\boldsymbol{\mathcal{H}}_{\Gamma\cup\Gamma_{1}}}-\delta_{tr}^{\boldsymbol{\mathfrak{M}}}\boldsymbol{s}_{\boldsymbol{\mathcal{G}}_{i}},\\
\boldsymbol{\mathcal{G}}_{i}&=&\sum\nolimits_{j=1}^{\kappa}\boldsymbol{\mathcal{U}}_{\boldsymbol{\mathcal{H}}}(:,j,:) \ast \boldsymbol{\mathcal{S}}_{\boldsymbol{\mathcal{G}}_{i}}(j,j,:) \ast \boldsymbol{\mathcal{V}}_{\boldsymbol{\mathcal{H}}}(:,j,:)^{*},\\
\boldsymbol{\mathcal{B}}_{i}&=&\sum\nolimits_{j=1}^{\kappa}\boldsymbol{\mathcal{U}}_{\boldsymbol{\mathcal{H}}}(:,j,:) \ast \boldsymbol{\mathcal{S}}_{\boldsymbol{\mathcal{B}}_{i}}(j,j,:) \ast \boldsymbol{\mathcal{V}}_{\boldsymbol{\mathcal{H}}}(:,j,:)^{*},\\
\boldsymbol{\mathcal{P}}_{i}&=&\sum\nolimits_{j=1}^{\kappa}\boldsymbol{\mathcal{U}}_{\boldsymbol{\mathcal{H}}}(:,j,:) \ast \boldsymbol{\mathcal{S}}_{\boldsymbol{\mathcal{P}}_{i}}(j,j,:) \ast \boldsymbol{\mathcal{V}}_{\boldsymbol{\mathcal{H}}}(:,j,:)^{*}.
\end{eqnarray*}
Then it is not hard to see that both $\boldsymbol{B}_{i}$ and $\boldsymbol{P}_{i}$ are all tensors with tubal rank at most $tr$ for $i=1,2,\cdots, N$, and
\begin{equation*}
 \boldsymbol{\mathcal{H}}_{\Gamma_{2}}=\sum\nolimits_{i=1}^{N}\gamma_{i}\boldsymbol{\mathcal{G}}_{i},~~\boldsymbol{\mathcal{B}}_{i}=(1+\delta_{tr}^{\boldsymbol{\mathfrak{M}}})\boldsymbol{\mathcal{H}}_{\Gamma\cup\Gamma_{1}}+\delta_{tr}^{\boldsymbol{\mathfrak{M}}}\boldsymbol{\mathcal{G}}_{i},~~\boldsymbol{\mathcal{P}}_{i}=(1-\delta_{tr}^{\boldsymbol{\mathfrak{M}}})\boldsymbol{\mathcal{H}}_{\Gamma\cup\Gamma_{1}}-\delta_{tr}^{\boldsymbol{\mathfrak{M}}}\boldsymbol{\mathcal{G}}_{i}.
\end{equation*}
Now we estimate the upper bounds of
\begin{equation*}
\xi\triangleq\sum\nolimits_{i=1}^{N}\gamma_{i}\left(\|\boldsymbol{\mathfrak{M}}(\boldsymbol{\mathcal{B}}_{i})\|_{2}^{2}-\|\boldsymbol{\mathfrak{M}}(\boldsymbol{\mathcal{P}}_{i})\|_{2}^{2}\right).
\end{equation*}
Applying Definition \ref{Tensor RIP Definition}, we have
\begin{eqnarray}\label{Lemma2-proof-inequality3}
\xi&=&4\delta_{tr}^{\boldsymbol{\mathfrak{M}}}\sum\nolimits_{i=1}^{N}\gamma_{i}\left\langle\boldsymbol{\mathfrak{M}}(\boldsymbol{\mathcal{H}}_{\Gamma\cup\Gamma_{1}}), \boldsymbol{\mathfrak{M}}(\boldsymbol{\mathcal{H}}_{\Gamma\cup\Gamma_{1}}+\boldsymbol{\mathcal{G}}_{i})\right\rangle\nonumber\\
&\stackrel{(a)}{=}&4\delta_{tr}^{\boldsymbol{\mathfrak{M}}}\left\langle\boldsymbol{\mathfrak{M}}(\boldsymbol{\mathcal{H}}_{\Gamma\cup\Gamma_{1}}), \boldsymbol{\mathfrak{M}}\left(\boldsymbol{\mathcal{H}}_{\Gamma\cup\Gamma_{1}}+\sum\nolimits_{i=1}^{N}\gamma_{i}\boldsymbol{\mathcal{G}}_{i}\right)\right\rangle\nonumber\\
&\stackrel{(b)}{=}&4\delta_{tr}^{\boldsymbol{\mathfrak{M}}}\langle\boldsymbol{\mathfrak{M}}(\boldsymbol{\mathcal{H}}_{\Gamma\cup\Gamma_{1}}), \boldsymbol{\mathfrak{M}}(\boldsymbol{\mathcal{H}})\rangle\nonumber\\
&\stackrel{(c)}{\leq}&4\delta_{tr}^{\boldsymbol{\mathfrak{M}}}\|\boldsymbol{\mathfrak{M}}(\boldsymbol{\mathcal{H}}_{\Gamma\cup\Gamma_{1}})\|_{2}\|\boldsymbol{\mathfrak{M}}(\boldsymbol{\mathcal{H}})\|_{2}\nonumber\\
&\stackrel{(d)}{\leq}&4\delta_{tr}^{\boldsymbol{\mathfrak{M}}}\sqrt{1+\delta_{tr}^{\boldsymbol{\mathfrak{M}}}}\|\boldsymbol{\mathcal{H}}_{\Gamma\cup\Gamma_{1}}\|_{F}\|\boldsymbol{\mathfrak{M}}(\boldsymbol{\mathcal{H}})\|_{2},
\end{eqnarray}
where (a) is due to $\sum_{i=1}^{N}\gamma_{i}=1$, (b) is founded on the fact that $\boldsymbol{\mathcal{H}}_{\Gamma_{2}}=\sum\nolimits_{i=1}^{N}\gamma_{i}\boldsymbol{\mathcal{G}}_{i}$ and $\boldsymbol{\mathcal{H}}=\boldsymbol{\mathcal{H}}_{\Gamma}+\boldsymbol{\mathcal{H}}_{\Gamma_{1}}+\boldsymbol{\mathcal{H}}_{\Gamma_{2}}$, (c) holds because of the Cauchy-Schwartz inequality, and (d) follows from (\ref{t-RIP}), $|\Gamma_{1}|<(t-1)r$ and the monotonicity of t-RIC.

Next, we use the block diagonal matrix to estimate the lower bound of $\xi$.
Let $\bar{\phi}\triangleq\|\bdiag(\boldsymbol{\bar{\mathcal{H}}}_{\Gamma^{c}})\|_{*}/(t-1)r$. Repeat step 1 for the matrix $\bdiag(\boldsymbol{\bar{\mathcal{H}}})$ as we did for tensor $\boldsymbol{\mathcal{H}}$ and we have
\begin{equation*}
 \sigma(\bdiag(\boldsymbol{\bar{\mathcal{H}}}_{\Gamma_{2}}))\in T(\bar{\phi}, (t-1)r-|E_{1}|),~~\boldsymbol{\bar{g}}_{i}\in U(\bar{\phi}, (t-1)r-|E_{1}|, \sigma(\bdiag(\boldsymbol{\bar{\mathcal{H}}}_{\Gamma_{2}}))),
\end{equation*}
here, $E_{1}$ is an index set as the counterpart of $\Gamma_{1}$. By further defining
\begin{eqnarray*}
\boldsymbol{\bar{b}}_{i}&=&(1+\delta_{tr}^{\boldsymbol{\mathfrak{M}}})\sigma(\bdiag(\boldsymbol{\bar{\mathcal{H}}}_{\Gamma\cup\Gamma_{1}}))+\delta_{tr}^{\boldsymbol{\mathfrak{M}}}\boldsymbol{\bar{g}}_{i},\\
\boldsymbol{\bar{p}}_{i}&=&(1-\delta_{tr}^{\boldsymbol{\mathfrak{M}}})\sigma(\bdiag(\boldsymbol{\bar{\mathcal{H}}}_{\Gamma\cup\Gamma_{1}}))-\delta_{tr}^{\boldsymbol{\mathfrak{M}}}\boldsymbol{\bar{g}}_{i},\\
\boldsymbol{\bar{G}}_{i}&=&\sum\nolimits_{j=1}^{\kappa}(\boldsymbol{u}_{\boldsymbol{\bar{H}}})_{j}\cdot(\boldsymbol{\bar{g}}_{i})_{j}\cdot(\boldsymbol{v}_{\boldsymbol{\bar{H}}})_{j}^{*},\\
\boldsymbol{\bar{B}}_{i}&=&\sum\nolimits_{j=1}^{\kappa}(\boldsymbol{u}_{\boldsymbol{\bar{H}}})_{j}\cdot(\boldsymbol{\bar{b}}_{i})_{j}\cdot(\boldsymbol{v}_{\boldsymbol{\bar{H}}})_{j}^{*},\\
\boldsymbol{\bar{P}}_{i}&=&\sum\nolimits_{j=1}^{\kappa}(\boldsymbol{u}_{\boldsymbol{\bar{H}}})_{j}\cdot(\boldsymbol{\bar{p}}_{i})_{j}\cdot(\boldsymbol{v}_{\boldsymbol{\bar{H}}})_{j}^{*}.
\end{eqnarray*}
Then we can easily induce that
\begin{eqnarray*}
\bdiag(\boldsymbol{\bar{\mathcal{H}}}_{\Gamma_{2}})&=&\sum\nolimits_{i=1}^{N}\gamma_{i}\boldsymbol{\bar{G}}_{i},\\
\boldsymbol{\bar{B}}_{i}&=&(1+\delta_{tr}^{\boldsymbol{\mathfrak{M}}})\bdiag(\boldsymbol{\bar{\mathcal{H}}}_{\Gamma\cup\Gamma_{1}})+\delta_{tr}^{\boldsymbol{\mathfrak{M}}}\boldsymbol{\bar{G}}_{i},\\
\boldsymbol{\bar{P}}_{i}&=&(1-\delta_{tr}^{\boldsymbol{\mathfrak{M}}})\bdiag(\boldsymbol{\bar{\mathcal{H}}}_{\Gamma\cup\Gamma_{1}})-\delta_{tr}^{\boldsymbol{\mathfrak{M}}}\boldsymbol{\bar{G}}_{i}.
\end{eqnarray*}
Thus, on the other hand, we also have
\begin{eqnarray}\label{Lemma2-proof-inequality4}
\xi&\stackrel{(a)}{\geq}&\sum\nolimits_{i=1}^{N}\gamma_{i}\left((1-\delta_{tr}^{\boldsymbol{\mathfrak{M}}})\|\boldsymbol{\mathcal{B}}_{i}\|_{F}^{2}-(1+\delta_{tr}^{\boldsymbol{\mathfrak{M}}})\|\boldsymbol{\mathcal{P}}_{i}\|_{F}^{2}\right)\nonumber\\
&\stackrel{(b)}{=}&\frac{1}{n_{3}}\sum\nolimits_{i=1}^{N}\gamma_{i}\left((1-\delta_{tr}^{\boldsymbol{\mathfrak{M}}})\|\boldsymbol{\bar{B}}_{i}\|_{2}^{2}-(1+\delta_{tr}^{\boldsymbol{\mathfrak{M}}})\|\boldsymbol{\bar{P}}_{i}\|_{2}^{2}\right)\nonumber\\
&\stackrel{(c)}{=}&\frac{2}{n_{3}}\delta_{tr}^{\boldsymbol{\mathfrak{M}}}(1-(\delta_{tr}^{\boldsymbol{\mathfrak{M}}})^2)\|\sigma(\bdiag(\boldsymbol{\bar{\mathcal{H}}}_{\Gamma\cup\Gamma_{1}}))\|_{2}^{2}-\frac{2}{n_{3}}(\delta_{tr}^{\boldsymbol{\mathfrak{M}}})^{3}\sum\nolimits_{i=1}^{N}\gamma_{i}\|\boldsymbol{\bar{g}}_{i}\|_{2}^{2}\nonumber\\
&\stackrel{(d)}{\geq}&\frac{2}{n_{3}}\delta_{tr}^{\boldsymbol{\mathfrak{M}}}(1-(\delta_{tr}^{\boldsymbol{\mathfrak{M}}})^2)\|\bdiag(\boldsymbol{\bar{\mathcal{H}}}_{\Gamma\cup\Gamma_{1}})\|_{F}^{2}-\frac{2(\delta_{tr}^{\boldsymbol{\mathfrak{M}}})^{3}}{n_{3}(t-1)r}\|\bdiag(\boldsymbol{\bar{\mathcal{H}}}_{\Gamma^{c}})\|_{*}^{2}\nonumber\\
&\stackrel{(e)}{=}&2\delta_{tr}^{\boldsymbol{\mathfrak{M}}}(1-(\delta_{tr}^{\boldsymbol{\mathfrak{M}}})^2)\|\boldsymbol{\mathcal{H}}_{\Gamma\cup\Gamma_{1}}\|_{F}^{2}-\frac{2n_{3}(\delta_{tr}^{\boldsymbol{\mathfrak{M}}})^{3}}{(t-1)r}\|\boldsymbol{\mathcal{H}}_{\Gamma^{c}}\|_{*}^{2},
\end{eqnarray}
where (a) follows from t-RIP, (b) holds because of (\ref{Property 1}), (c) is due to  $\langle\sigma(\bdiag(\boldsymbol{\bar{\mathcal{H}}}_{\Gamma\cup\Gamma_{1}})), \boldsymbol{\bar{g}}_{i}\rangle=0$ for all $i=1,2,\cdots, N$, (d) is based on the fact that $\|\boldsymbol{X}\|_{F}=\|\sigma(\boldsymbol{X})\|_{2}$ for any matrix $\boldsymbol{X}$ and
\begin{equation*}
\|\boldsymbol{\bar{g}}_{i}\|_{2}^{2}\leq\|\boldsymbol{\bar{g}}_{i}\|_{0}(\|\boldsymbol{\bar{g}}_{i}\|_{\infty})^{2}\leq((t-1)r-|E_{1}|)\bar{\phi}^{2}\leq\frac{\|\bdiag(\boldsymbol{\bar{\mathcal{H}}}_{\Gamma^{c}})\|_{*}^{2}}{(t-1)r},
\end{equation*}
and (e) follows from (\ref{Property 2}).

Combining (\ref{Lemma2-proof-inequality3}) and (\ref{Lemma2-proof-inequality4}), we get
\begin{equation}\label{Lemma2-proof-inequality5}
 (1-(\delta_{tr}^{\boldsymbol{\mathfrak{M}}})^2)\|\boldsymbol{\mathcal{H}}_{\Gamma\cup\Gamma_{1}}\|_{F}^{2}-\frac{n_{3}(\delta_{tr}^{\boldsymbol{\mathfrak{M}}})^{2}}{(t-1)r}\|\boldsymbol{\mathcal{H}}_{\Gamma^{c}}\|_{*}^{2}
\leq2\sqrt{1+\delta_{tr}^{\boldsymbol{\mathfrak{M}}}}\|\boldsymbol{\mathcal{H}}_{\Gamma\cup\Gamma_{1}}\|_{F}\|\boldsymbol{\mathfrak{M}}(\boldsymbol{\mathcal{H}})\|_{2}.
\end{equation}
Obviously, (\ref{Lemma2-proof-inequality5}) is a quadratic inequality in terms of $\|\boldsymbol{\mathcal{H}}_{\Gamma\cup\Gamma_{1}}\|_{F}$. Using extract roots formula, we obtain
\begin{eqnarray*}
&&\|\boldsymbol{\mathcal{H}}_{\Gamma\cup\Gamma_{1}}\|_{F}\\
 &\leq&\frac{2\sqrt{1+\delta_{tr}^{\boldsymbol{\mathfrak{M}}}}\|\boldsymbol{\mathfrak{M}}(\boldsymbol{\mathcal{H}})\|_{2}+\sqrt{(2\sqrt{1+\delta_{tr}^{\boldsymbol{\mathfrak{M}}}}\|\boldsymbol{\mathfrak{M}}(\boldsymbol{\mathcal{H}})\|_{2})^{2}+4(1-
(\delta_{tr}^{\boldsymbol{\mathfrak{M}}})^{2})\frac{n_{3}(\delta_{tr}^{\boldsymbol{\mathfrak{M}}})^{2}}{(t-1)r}\|\boldsymbol{\mathcal{H}}_{\Gamma^{c}}\|_{*}^{2}}}{2(1-(\delta_{tr}^{\boldsymbol{\mathfrak{M}}})^2)}\\
&\leq&\frac{2}{(1-\delta_{tr}^{\boldsymbol{\mathfrak{M}}})\sqrt{1+\delta_{tr}^{\boldsymbol{\mathfrak{M}}}}}\|\boldsymbol{\mathfrak{M}}(\boldsymbol{\mathcal{H}})\|_{2}+\frac{\sqrt{n_{3}}\delta_{tr}^{\boldsymbol{\mathfrak{M}}}}{\sqrt{(1-(\delta_{tr}^{\boldsymbol{\mathfrak{M}}})^{2})(t-1)}}
\frac{\|\boldsymbol{\mathcal{H}}_{\Gamma^{c}}\|_{*}}{\sqrt{r}},
\end{eqnarray*}
where the last inequality is based on the fact that $\sqrt{x^{2}+y^{2}}\leq|x|+|y|$. Therefore we prove (\ref{Lemma2-proof-inequality1}). Since we also have
$\|\boldsymbol{\mathcal{H}}\|_{F}\leq\|\boldsymbol{\mathcal{H}}_{\Gamma\cup\Gamma_{1}\|_{F}}$, it is easy to induce (\ref{Lemma2-results}).

When $tr$ is not an integer, let $t^{\prime}=\lceil tr\rceil/r$, then $t^{\prime}r$ is an integer. Based on the definition of t-RIP and $tr < \lceil tr\rceil < tr + 1$, we have $\delta_{t^{\prime}r}=\delta_{tr}<1$. Thus (\ref{Lemma2-results}) holds no matter whether $tr$ is an integer or not. This completes the proof.
\end{proof}

\end{document}